\documentclass[11pt]{article}
\usepackage[margin=1in]{geometry}
\usepackage{microtype}

\usepackage{amsmath,amssymb, graphicx,wrapfig, url}
\usepackage{color}

\usepackage{algorithm2e}
\usepackage{hyperref}

\usepackage{amsthm}
\newtheorem{lemma}{Lemma}
\newtheorem{theorem}{Theorem}
\newtheorem{property}{Property}

\newcommand{\col}{\mathcal{K}}
\newcommand{\C}{\mathcal{C}}
\newcommand{\0}[1]{}   

\newtheorem{cor}{Corollary}
\newtheorem{observation}{Observation}

\newcommand{\figref}[1]{\autoref{fig:#1}}

\begin{document}
\date{}

\title{Minimum Forcing Sets for Miura Folding Patterns}
\author{Brad Ballinger\thanks{Humboldt State University, Arcata, USA, \texttt{bradley.ballinger@humboldt.edu}} \and Mirela Damian\thanks{Computer Science Dept., Villanova University, Villanova, USA, \texttt{mirela.damian@villanova.edu}} \and David Eppstein\thanks{Computer Science Dept., Univ. of California, Irvine, USA, \texttt{david.eppstein@gmail.com}; work supported by the National Science Foundation under Grant CCF-1228639 and by the Office of Naval Research under Grant No. N00014-08-1-1015.} \and Robin Flatland\thanks{Computer Science Dept., Siena College, Loudonville, NY, USA, \texttt{flatland@siena.edu}} \and Jessica Ginepro\thanks{Department of Mathematics, University of Connecticut, \texttt{jessginepro@yahoo.com}} \and Thomas Hull\thanks{Department of Mathematics, Western New England University, \texttt{thull@wne.edu}; work supported by the National Science Foundation under grant EFRI-1240441.}}
    
\maketitle

\begin{abstract}
We introduce the study of \emph{forcing sets} in mathematical origami. The origami material folds flat along straight line segments called \emph{creases}, each of which is assigned a folding direction of mountain or valley. A subset $F$ of creases is \emph{forcing} if the global folding mountain/valley assignment can be deduced from its restriction to $F$. 
In this paper we focus on one particular class of foldable patterns called Miura-ori, which divide the plane into congruent parallelograms using horizontal lines and zig-zag vertical lines.  
We develop efficient algorithms for constructing a minimum forcing set of a Miura-ori map,
and for deciding whether a given set of creases is forcing or not. We also provide tight bounds on the size of a forcing set, establishing that the standard mountain-valley assignment for the Miura-ori is the one that requires the most creases in its forcing sets. Additionally, given a partial mountain/valley assignment to a subset of creases of a Miura-ori map, we determine  whether the assignment domain can be extended to a locally flat-foldable pattern on all the creases. At the heart of our results is a novel correspondence between flat-foldable Miura-ori maps and $3$-colorings of grid graphs. 
\end{abstract}


\section{Introduction}

In the mathematical modeling of programmable matter, the topic of \emph{self-folding origami}---where a pre-programmed material folds itself from an initially flat state in response to some stimulus or mechanism---has been gaining in popularity.  The origami material can be pre-programmed to fold along straight lines, called \emph{creases}, by means of rotation to achieve a mountain or valley fold (for examples, see~\cite{Hawkes+10, Ionov13, Leong+07,MR+05}). In such cases, the self-folding process can be economized by programming only a subset of the creases to self-fold, which would then \emph{force} the other, passive creases to fold as originally intended.  We call such a subset of the creases a {\em forcing set}.\0{I think there's a difference between choosing a subset $F$ that is logically compatible with only one global folding pattern (which is what we've done), and choosing a subset $F'$ that can actually encourage that global pattern (which is what we talk about here); in practice, our forcing set might nudge a crease pattern to nondeterministically choose some local folding that doesn't end well.}

Because of the cost of programming creases, finding a forcing set of smallest size in a given crease pattern would be a useful tool for self-folding. To our knowledge, this is the first paper to address this problem in two dimensions.\footnote{For one dimensional crease patterns, there are some unpublished forcing set results by a subset of the authors here.}
Due both to the high computational complexity of dealing with unrestricted crease patterns~\cite{BerHay-SODA-96}, and to the versatility of modular repeating crease patterns in the design of self-folding shapes~\cite{ABDR11},
our focus here is on finding forcing sets for flat-foldable origami whose folded shape lies flat in the plane, using a restricted set of folds based on the \emph{Miura-ori} pattern
(see \figref{StdMiura}a). This is a common crease pattern in origami that has found applications ranging from tourist maps to solar panel arrays for satellites~\cite{Hull06}.

In this paper,
we develop efficient algorithms for determining a minimum forcing set of a Miura-ori map, and for deciding whether a given set of creases is forcing or not. Given a partial mountain/valley assignment to a subset of creases of a Miura-ori map, we also determine  whether the assignment domain can be extended to cover all creases such that the resulting pattern is locally flat-foldable.
Additionally, we show that, for a Miura-ori crease pattern with $m \times n$ cells and a flat-foldable assignment of mountain/valley folds to its creases, a forcing set includes at least $m+n-2$ creases and at most $\lceil mn/2 \rceil$ creases. We establish a lower bound of $mn/2$ on the size of a forcing set for the standard Miura-ori mountain-valley assignment (described in \autoref{sec:definitions}), so from this perspective the standard  Miura-ori is the ``worst'' Miura.  
These results make use of previous work characterizing when a crease pattern at a node 
is locally flat-foldable~\cite{Hull03CountingMV} (detailed in \autoref{sec:localfold}) and a novel correspondence between flat-foldable Miura-ori maps and $3$-colorings of grid graphs recently developed in~\cite{HullGinepro} (detailed in \autoref{sec:gridgraph}).

Although we are the first to consider the minimum forcing set problem,
there is much related work on flat-foldable origami. 
Arkin et al.~\cite{Arkin04} give a linear time algorithm for determining if
a one dimensional crease pattern is flat-foldable.  In~\cite{BerHay-SODA-96}, Bern and Hayes show how to determine in linear time whether a general
crease pattern has a mountain-valley assignment for which every node 
is locally flat-foldable. This is necessary but not sufficient to ensure that the entire crease pattern is flat-foldable. Deciding whether a crease pattern is flat-foldable is NP-hard~\cite{BerHay-SODA-96}.  For crease patterns consisting of a regular $m \times n$ grid of squares, the complexity of deciding whether a given mountain-valley assignment can be folded flat (the \emph{map folding problem}) remains open~\cite{DOR07}, although recent progress has been made on $2 \times n$ grids~\cite{Morgan12} and on testing for valid
linear orderings of the faces in a flat-folding~\cite{NW13}. In~\cite{Hull03CountingMV}, Hull gives upper and lower bounds on the number of flat-foldable mountain-valley assignments on a single-node 
crease pattern.

\subsection{Definitions}
\label{sec:definitions}

\0{Define ``flat origami'', or include reference to the appropriate section.  Maybe our readers will know before our definition in line 53 that 0 is not in the range of $\mu$---that when something is ``unfolded back to a piece of paper'', it ceases to be flat origami...but we could make that clearer anyway.}
A \emph{crease pattern} is a finite planar straight-line graph drawn on a piece of paper, with \emph{edges}
corresponding to line segments of the crease pattern, and \emph{nodes} connecting two or more edges. 
For clarity, we will refer to the edges of the planar graph simply as \emph{creases}. 
A crease pattern is \emph{flat-foldable} if it is the crease pattern of some flat origami; that is, if there is a way of folding the paper so that the folds lie along the creases and the folded shape again lies flat in a plane~\cite{BerHay-SODA-96}.

For a given crease pattern $C$, let $E(C)$ denote the set of creases 
of the corresponding planar graph.  
A {\em mountain-valley assignment} on $C$, or MV \emph{assignment} for short, is a function $\mu: E(C)  \rightarrow \{-1, +1\}$, where $-1$ indicates a valley crease and $+1$ indicates a mountain crease.  An MV assignment restricts the way a crease pattern folds into some flat origami. Intuitively, if the origami is unfolded back to a piece of paper with one side up and one side down, a mountain fold points up and a valley fold points down. Given a foldable MV assignment $\mu$, a flat folding that conforms to $\mu$ is not necessarily uniquely determined, i.e., there may be different ways to fold, and different flat states in terms of paper stacking and tucking, all of which respect $\mu$.
In this work we will not be paying attention to different layering states of the paper. 
An MV assignment that can be folded to form a flat origami is called \emph{foldable}. 
In order for the crease pattern to fold flat, the function $\mu$ must obey certain rules, which will be discussed in \autoref{sec:localfold}. A crease pattern $C$ along with an assignment $\mu$ on $C$ form a \emph{folding pattern} $(C, \mu)$. 

Given a folding pattern ($C, \mu$), we say that a subset $F$ of $C$ is {\em forcing} if the only flat-foldable MV assignment on $C$ that agrees with $\mu$ on $F$ is $\mu$ itself. This means that, if each crease $a \in F$ is assigned the value $\mu(a)$, then each crease $b \in C \setminus F$ must be assigned the value $\mu(b)$, in order to produce a foldable MV assignment. 
A forcing set $F$ is called {\em minimum} if there is no other forcing set with size less than $|F|$.  

An $m \times n$ Miura-ori crease pattern (see \figref{StdMiura}a) refers to a rectangular sheet of paper divided by equidistant parallel (\emph{e.g.}, horizontal) creases into $m$ strips, each of which is further divided by equidistant parallel creases (oblique to the first set---\emph{e.g.}, nearly vertical) into $n$ quadrilaterals, such that the crease pattern on one strip is the mirror image of the crease pattern on any adjacent strip.  
Except at the strip ends, all cells are congruent parallelograms containing a pair of acute node 
angles $\alpha$ and obtuse node 
angles $\pi-\alpha$.  
In the \emph{standard} Miura-ori MV assignment shown in \figref{StdMiura}a, each zig-zag crease path is monochrome (all mountain creases or all valley creases), but adjacent zig-zags are of the opposite orientation.  
Each straight crease path alternates: a mountain crease; a valley crease; repeat. 
 
\begin{figure}[tbp]
\centering
\begin{tabular}{c@{\hspace{0.02in}}c@{\hspace{0.02in}}c}
\includegraphics[width=0.38\linewidth]{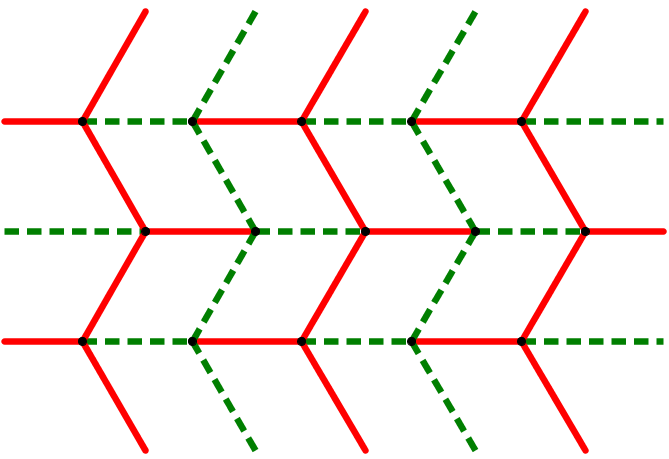}\0{StandardMiura} & 
\raisebox{0.0in}{\includegraphics[width=0.3\linewidth]{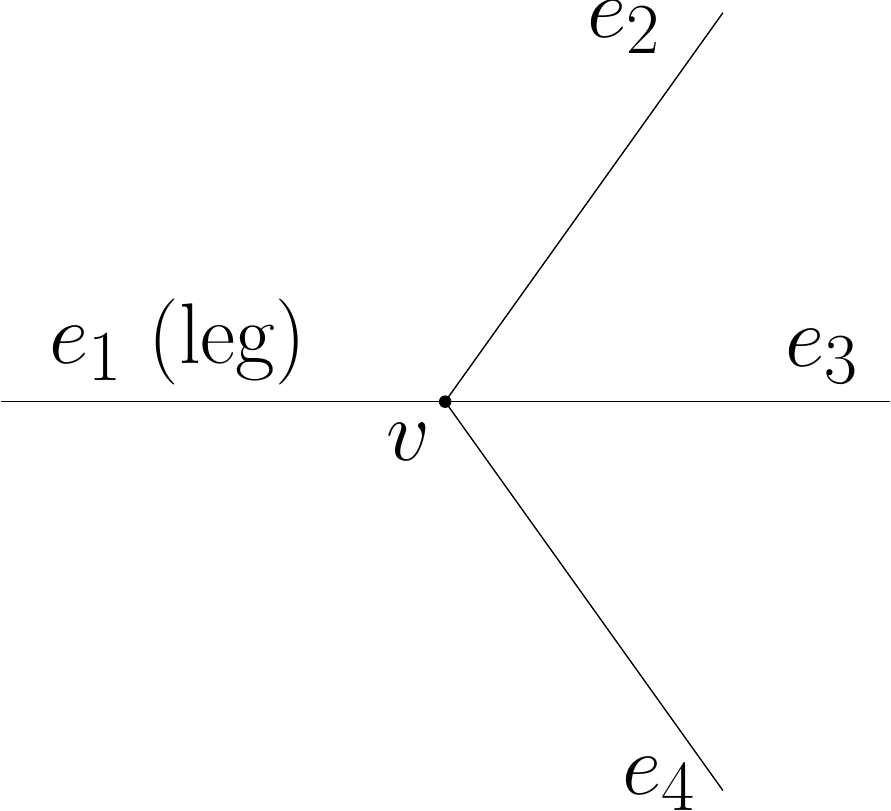}} &
\raisebox{0.05in}{\includegraphics[width=0.28\linewidth]{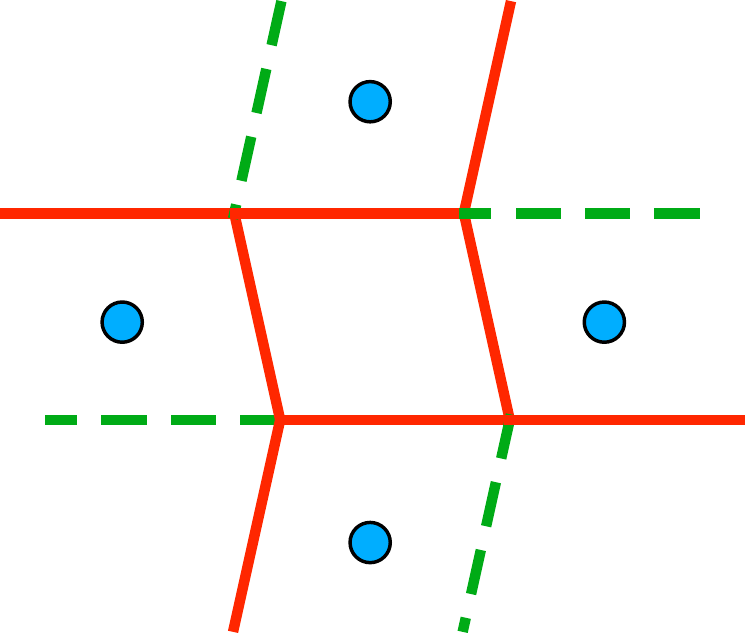}} \\
(a) & (b) & (c)
\end{tabular}
\caption{(a) A $4 \times 6$ Miura map with the standard MV assignment.  Solid red creases 
represent mountain folds, while dashed green creases represent valley folds. (b) Creases incident on a single node $v$. (c) A non-standard MV assignment to a Miura-ori that is flat-foldable at each node but not globally.  Assuming the center cell is the top layer, each cell marked with a blue point must be below the (clockwise) next one.}
\label{fig:StdMiura}
\end{figure}

\section{Local Flat Foldability}
\label{sec:localfold}

We begin with some results from the combinatorial origami literature that will prove useful in our discussion.  More details can be found in \cite{Hull03CountingMV}.

\begin{theorem}[Maekawa's Theorem]
At any node in a flat-foldable crease pattern, the difference between the number of mountain and valley creases is always two.
\end{theorem}

For the Miura-ori, this means that each node $v$ will have either 3 mountains and 1 valley or vice-versa.  However, the geometry of these nodes impose further structure.  As noted, there are four creases incident to $v$; exactly one of these, which we designate $e_1$---and affectionately call  the \emph{bird leg}, or just \emph{leg}---is separated by obtuse angles from its nearest neighbors. 
We call the others collectively \emph{toes}, but individually $e_2$, $e_3$, and $e_4$, so that $e_3$ forms a straight angle with $e_1$, as seen in \figref{StdMiura}b.
In this sense, $e_3$ is the middle toe, while $e_2$ and $e_4$ are the lateral toes.
Notice that by Maekawa's Theorem, we have $|\sum_i \mu(e_i)|=2$. 

\begin{theorem}[Bird's Foot Theorem]\label{birds foot}
Using the notation in \figref{StdMiura}b, any Miura-ori flat-foldable node must have $\mu(e_1)=\sum_{i=2}^4 \mu(e_i)$.
\end{theorem}

In other words, the Bird's Foot Theorem states that in any Miura-ori flat-foldable node, the leg of the foot (crease $e_1$) must have the same MV parity as the majority of the toes (creases $e_2$, $e_3$, and $e_4$).  The reason for this is that the only other possibility (via Maekawa's Theorem) would be for, say, $e_1$ to be a valley and all the other creases to be mountains. This is impossible to fold flat because it would require the acute node angles $\alpha$ to contain the obtuse node angles $\pi-\alpha$, which would require the paper to crumple or tear; see \figref{BirdFoot}.

\begin{figure}[tbp]
\centering
\begin{tabular}{c@{\hspace{0.05in}}c}
\includegraphics[width=0.4\linewidth]{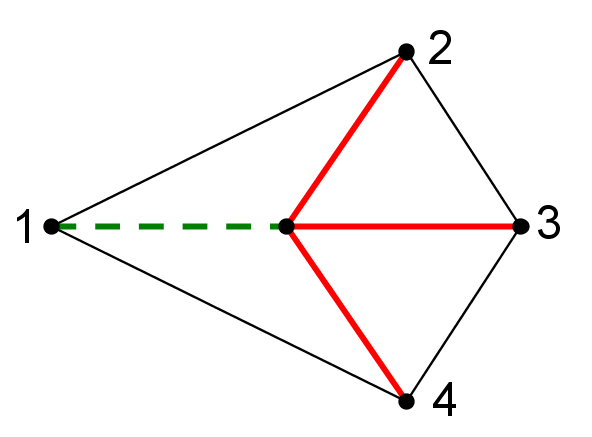} & 
\raisebox{0.2in}{\includegraphics[width=0.55\linewidth]{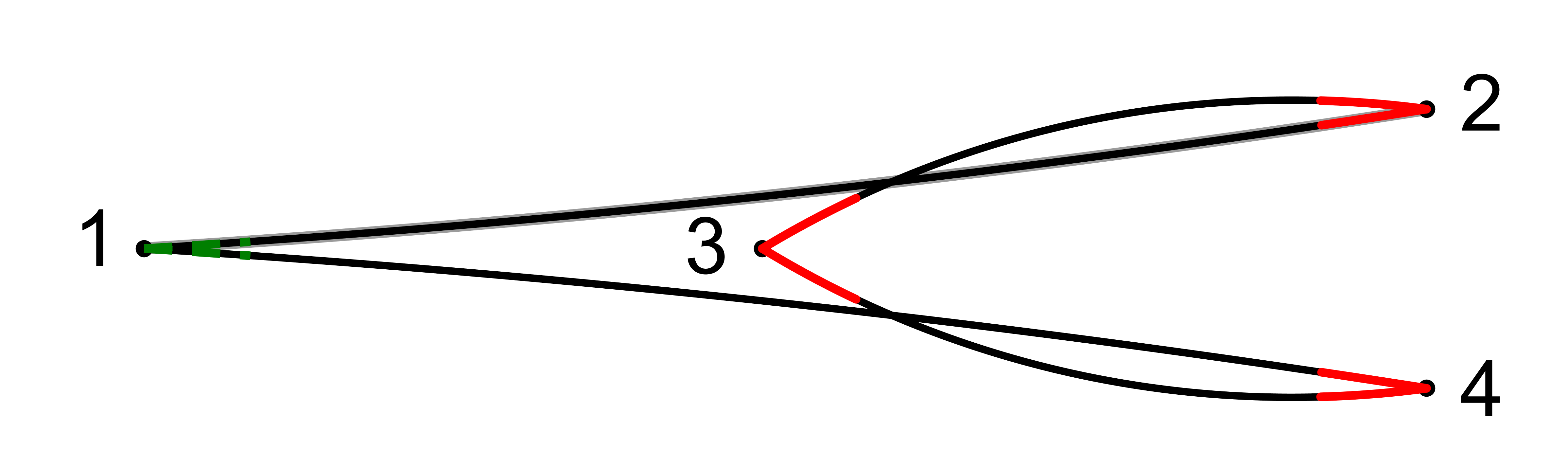}} \\
(a) & (b) 
\end{tabular}
\caption{To illustrate the contradiction proving the Bird's Foot Theorem: (a) A quadrilateral cut out of a Miura map with a valley at $e_1$ and mountains at $e_2$, $e_3$, and $e_4$. (b) When we fold this flat and point its central node away from us, we see its mountain folds as convex turns and its valley fold as a concave turn.  Now the pair of long creases collide with the short creases.}
\label{fig:BirdFoot}
\end{figure}
%
With the standard MV assignment in particular, the middle toe is the odd one out: $\mu(e_1)=\mu(e_2)=\mu(e_4)=-\mu(e_3)$. 

\0{\begin{cor}
\label{cor:stdflip}
If $\mu$ is a locally flat-foldable MV assignment, then at each node we have $\mu(e_2)=\mu(e_4)$ if and only if $\mu(e_1)=-\mu(e_3)$.  In either case, $\mu(e_1)=\mu(e_2)=-\mu(e_3)=\mu(e_4)$.
\end{cor}

For the forward implication, when the lateral toes agree, they constitute a majority of the toes, so the leg must agree with them and the middle toe must not. 
For the reverse implication, when the leg disagrees with the middle toe, it must agree with both lateral toes in order to agree with a majority of the toes. }

\begin{cor}
\label{cor:stdflip}
If $\mu$ is a locally flat-foldable MV assignment, then at any node where either $\mu(e_2)=\mu(e_4)$ or $\mu(e_1)=-\mu(e_3)$, we have $\mu(e_1)=\mu(e_2)=-\mu(e_3)=\mu(e_4)$.
\end{cor}

Note that a Miura-ori that is locally flat-foldable at each node under a MV assignment $\mu$ is not necessarily
globally flat-foldable under $\mu$. One example supporting this claim is depicted in \figref{StdMiura}c.

%

\section{Equivalence to 3-Colorings of Grid Graphs}
\label{sec:gridgraph}

A key tool for what follows is the existence of a bijection between the locally flat-foldable Miura-ori MV assignments and the vertex 3-colorings of grid graphs with one vertex pre-colored.  We explain this equivalence here and refer the reader to \cite{HullGinepro} for the proof.

\begin{figure}[htbp]
\centering
\includegraphics[width=0.6\linewidth]{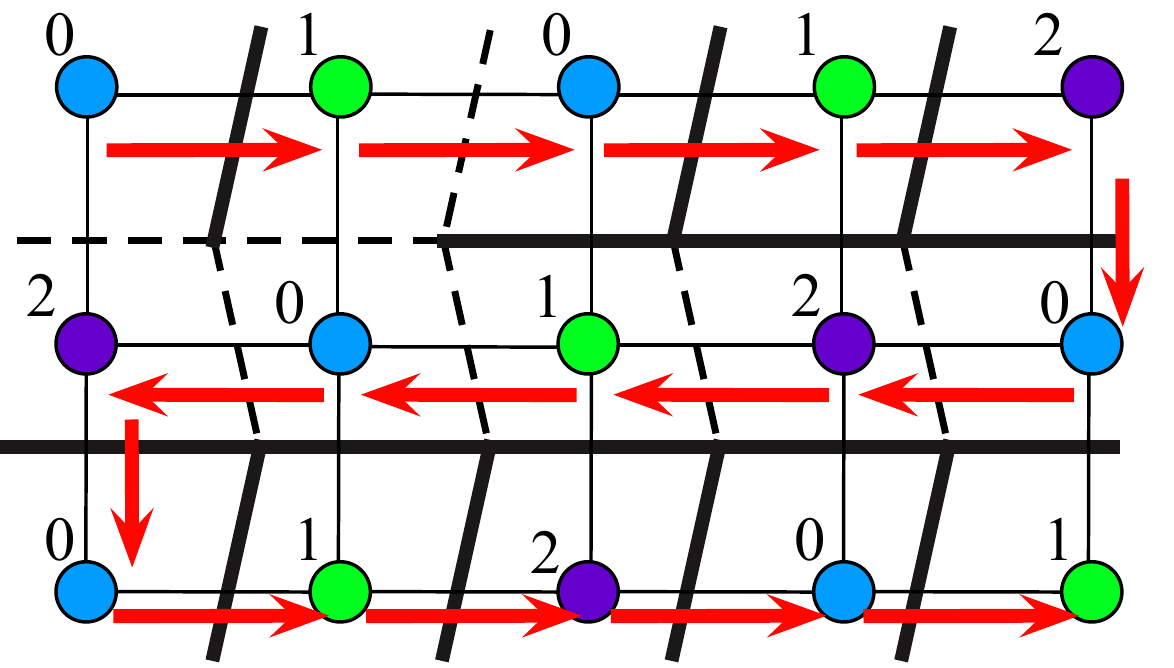}  
\caption{An example of the equivalence between locally flat-foldable Miura-ori MV assignments and 3-vertex colorings of a grid graph with one vertex pre-colored.}
\label{fig:StdMiura2}
\end{figure}


Think of a Miura-ori crease pattern as a plane graph and take its planar dual, $G$, 
ignoring the outside face.  Then $G$ is 
a grid graph with one vertex in each cell\0{parallelogram} of the crease pattern, 
and edges connecting vertices in adjacent cells. 
We orient our Miura-ori crease pattern so that the top row of nodes have their bird legs pointing left.  Suppose we are given a locally flat-foldable MV assignment $\mu$.  We use this to generate  a 3-vertex coloring $\col$ of $G$ as follows.  Color the upper-left vertex with color 0.  Then travel along the edges of $G$ starting at the upper-left vertex, traveling to the right along the top row, then going down one vertex, then traveling to the left along the second row, then down one vertex, then to the right again and so on.  See \figref{StdMiura2}, where this path is marked by the red arrows.  We use this path to construct the 3-coloring $\col$ recursively;  assume vertex $v_i$ on the path has been given color $\col(v_i)\in \{0,1,2\}$.  Let $v_{i+1}$ be the next vertex on the path and let $L_i$ be the crease of the Miura-ori between $v_i$ and $v_{i+1}$.  Then define $v_{i+1}$'s color using $\col(v_{i+1}) = \col(v_i)+\mu(L_i)\mod 3.$ 
An example of this correspondence is shown in \figref{StdMiura2}.  Furthermore, the process is reversible:  given a 3-vertex coloring of the graph $G$ with the upper-level vertex colored 0, we can generate a MV assignment $\mu$ for creases $L_i$ between consecutive vertices $v_i$ and $v_{i+1}$ on the red path using
$$\mu(L_i) = \left\{ \begin{array}{cl}
1 & \mbox{if }\col(v_{i+1})-\col(v_{i})=1\mod 3 \\
-1 & \mbox{if }\col(v_{i+1})-\col(v_{i})=2\mod 3.\end{array}\right.$$ 
This extends uniquely to a full MV assignment $\mu$ which will be locally flat-foldable.  See \cite{HullGinepro} for details.  

This correspondence between the MV assignment $\mu$ and the coloring $\col$ gives us a mapping between a forcing set $F$ of the Miura-ori pattern relative to $\mu$ and a {\em forcing subset of edges $S$ of $G$ relative to $\col$}. The subset $S$ would thus have the property that the only 3-coloring of the vertices of $G$ that agrees with $\col$ on the vertices incident to $S$ is $\col$ itself. 

\section{Tight Bounds for Standard Miura-ori}
\label{sec:standardmiura}
In this section we are concerned with minimum forcing sets for the standard MV assignment on a Miura-ori such as the one shown in \figref{StdMiura}a.  
%
%
%
%
It is helpful to think of the Miura-ori as a perturbed rectangular array of unit squares, which we will now tile with $2 \times 1$ rectangles (dominoes). 

\begin{figure}[htbp]
\centering
\begin{tabular}{c@{\hspace{0.1in}}c}
\includegraphics[width=0.45\linewidth]{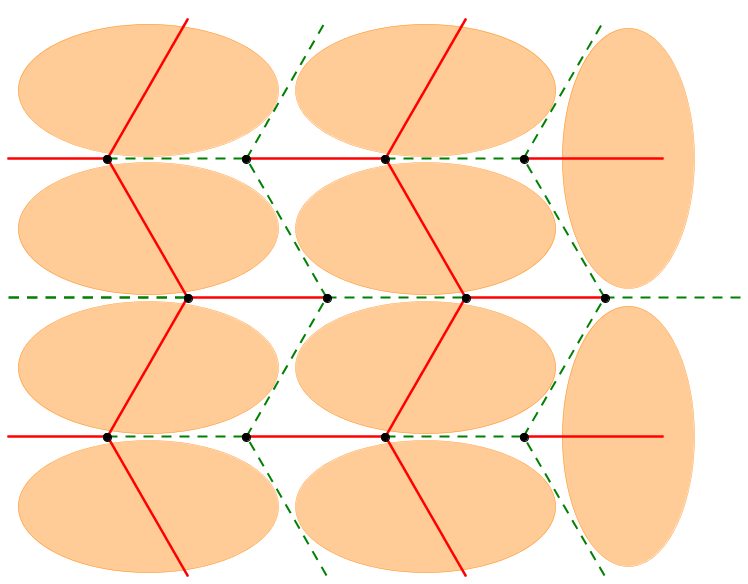} & 
\includegraphics[width=0.45\linewidth]{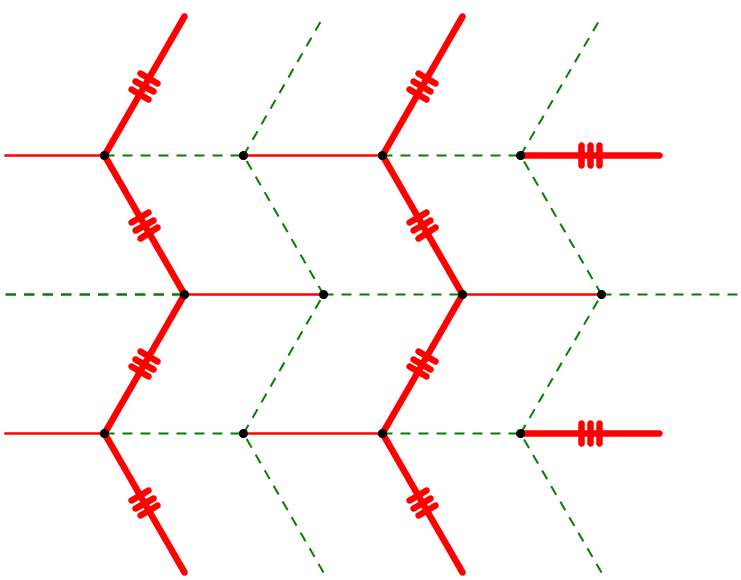} \\
(a) & (b) 
\end{tabular}
\caption{(a) A tiling by dominoes rendered as ellipses.  (b) The forcing set for the tiling from (b) includes the marked creases.}
\label{fig:domino}
\end{figure}


We assume at first that the Miura-ori has an even number of columns---which is to say, an odd number of zig-zags. 
We tile it with horizontal dominoes (see the left four columns of \figref{domino}a).
Each domino crosses exactly one crease of the Miura-ori.  
With this particular tiling and MV assignment, we see that every mountain zig-zag crease path (and no other crease) is completely covered by dominoes; let $F$ consist of these creases. 
By \autoref{cor:stdflip}, $F$ is a forcing set for $\mu$.

Now suppose that the Miura-ori has an odd number of columns, but an even number of rows. 
We tile the rightmost column (``column $n$'') with vertical dominoes, and the rest of the grid as before. 
We have again covered every mountain zig-zag crease path, and also the horizontal mountain creases (middle toes) in column $n$. 
With this choice of $F$, we see that $\mu$ is determined for all creases except, perhaps, the valley creases in column $n$ and its bounding zig-zag. 
Since $\mu$ is forced in column $n-1$, the creases that appear there as valleys are indeed forced to be valleys. 
These give us a valley leg and mountain toes at alternate nodes adjacent to column $n$---which, by \autoref{cor:stdflip}, forces $\mu$ on the rightmost valley zig-zag and therefore throughout the map.

When the Miura-ori has an odd number of rows and columns, we omit the bottom-right cell from the tiling, but otherwise proceed as above---tiling column $n$ with vertical dominoes, and the rest with horizontal dominoes. 
The collection of mountain creases corresponding to this domino tiling is sufficient to force $\mu$ everywhere except for the two creases bounding the bottom-right cell.  Around their common node, we know that one toe is a mountain crease and another a valley, but now flat-foldability requires only that the remaining toe and leg agree with each other: they could be both mountains or both valleys.  Thus, we add one of these creases to $F$ in order to force the standard MV assignment across the grid.

Suppose that $F$ is a forcing set corresponding to a domino tiling in which some $2 \times 2$ square is covered by a pair of parallel dominoes. 
Let $F'$ be the crease set corresponding to the domino tiling that results when these two parallel dominoes are flipped: turn a horizontal pair vertical, or vice-versa. 
\autoref{cor:stdflip} implies that $F'$ is also a forcing set. 
Furthermore, all domino tilings of the Miura grid are flip-connected~\cite{S+95}, so any domino tiling of the grid (in which $mn$ is even) corresponds to a forcing set of size $\frac{mn}{2}$. 
Thus we have the following result. 


\begin{theorem}
When $mn$ is even, an $m \times n$ Miura-ori crease pattern with the standard MV assignment has a minimum forcing set of $\frac{mn}{2}$ creases.
\end{theorem}

\begin{proof}
The domino tiling procedure presented above produces a forcing set of $\frac{mn}{2}$ creases. To
show that this is of minimum size, it is easily verified that negating the standard MV assignment along the boundary of a single cell (turning two mountains into valleys and two valleys into mountains, for an interior cell) produces a new MV assignment that satisfies conditions for local flat-foldability. 
Therefore, a set $F \subset E(C)$ must contain at least one crease of each cell in order to force $\mu$.  As a crease may serve this purpose in up to two cells, a forcing set must contain at least $\frac{mn}{2}$ creases.  This implies that all of the forcing sets based on domino tilings are minimum for the standard MV-assignment.
\end{proof}



For non-standard MV assignments, domino tilings are not guaranteed to produce forcing sets.
For example, consider a single bird's foot with the flat-foldable assignment
$\mu(e_1) = -1, \mu(e_2) = +1, \mu(e_3) = -1$, and $\mu(e_4) = -1$. A
domino tiling covers creases $e_2,e_4$, but fixing the MV assignments of these two creases
does not force the assignments of the remaining two creases; the other two creases could 
both be +1 or they could both be -1 for flat-foldability.  
Therefore, in the remainder of the paper we consider the more challenging problem of finding forcing sets for an arbitrary flat-foldable MV-assignment.

\section{Tight Bounds for Non-Standard Miura-ori}
\label{sec:nonstandard}
Throughout this section, we view an $m \times n$ 3-colored grid graph as an $m \times n$ grid of squares with colored vertices placed at their centers. We maintain the term \emph{edge} to refer to an edge connecting two adjacent square centers (see the dashed edges in \figref{uniform}a), but distinguish it from an \emph{arc}, which refers to an edge of the square grid (see the solid edges in \figref{uniform}a). Note that each edge crosses exactly one arc. We imagine each square having the same color as its center. Let $\col$ be a 3-coloring of a given grid graph $G$.  For simplicity, 
we define $\col(s)$ for a square $s$ to be the same as $\col(v)$, where $v$ is the center of~$s$. 

\begin{figure*}[htbp]
\centering
\includegraphics[width=\linewidth]{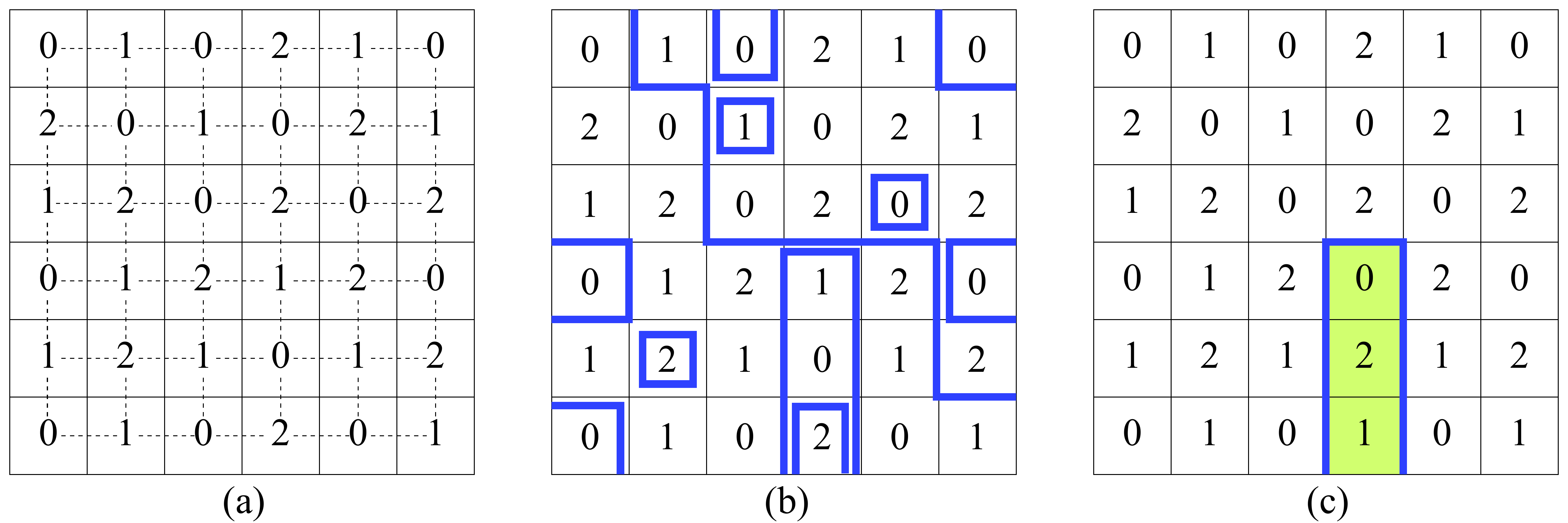}
\caption{Definitions (a) 3-colored grid: dashed segments are \emph{edges}; solid segments are \emph{arcs}; the colors are 0, 1, and 2 (b) Uniform curves marked by thick lines (c) Valid coloring obtained by adding $-1\pmod{3}$ to all colors on one side of the marked curve.}
\label{fig:uniform}
\end{figure*}

A \emph{uniform curve} $\C$ of $G$ is a boundary-to-boundary path or cycle of interior arcs with the property that the color difference (mod 3) is the same across every arc of $\C$. In other words, for any two edges $(a_1, b_1)$ and $(a_2, b_2)$ crossed by $\C$, with $a_1$ and $a_2$ on the same side of $C$, the equality 
$\col(b_1)-\col(a_1)=\col(b_2)-\col(a_2)$ holds.  \figref{uniform}b shows several\0{{the set of} the set is not fully represented; see def-uniform-piece-missing for some other uniform curves} uniform curves for the grid graph from \figref{uniform}a. The term ``uniform'' is used to remind the reader that the color difference is uniform across each arc of $\C$. As we shall later see, there is a one-to-one correspondence between the set of forcing edges and the set of uniform curves in $G$.
The following property follows immediately from the definition of a uniform curve. 

\begin{property}
\label{prop:uniform}
Let $x$ and $y$ be the two colors across an arc of a uniform curve $\C$, and let $d=x-y\in \{-1, +1\}$.  Now adding either $d\pmod{3}$ to all colors on the same side of $\C$ as $x$, or $-d\pmod{3}$ to all colors on the same side of $\C$ as $y$, yields another valid coloring.
\end{property}

\begin{cor}
\label{cor:uniform}
Any forcing set for $G$ must include an edge across every uniform curve in $G$.
\end{cor}
If no edge in a forcing set $F$ cuts across a uniform curve $\C$, then we can modify the coloring by adding $+1\pmod 3$ or $-1\pmod 3$ to all colors on one side of $\C$ and get a different valid coloring  (by \autoref{prop:uniform}), contradicting the fact that $F$ is forcing.

\subsection{Minimum Forcing Set Algorithm}
\label{sec:alg}
Before describing an algorithm for finding a minimum forcing set of a 3-colored grid graph $G$, 
we provide some intuition behind our approach to force a particular coloring on $G$. 

\begin{figure}[htbp]
\centering
\includegraphics[width=\linewidth]{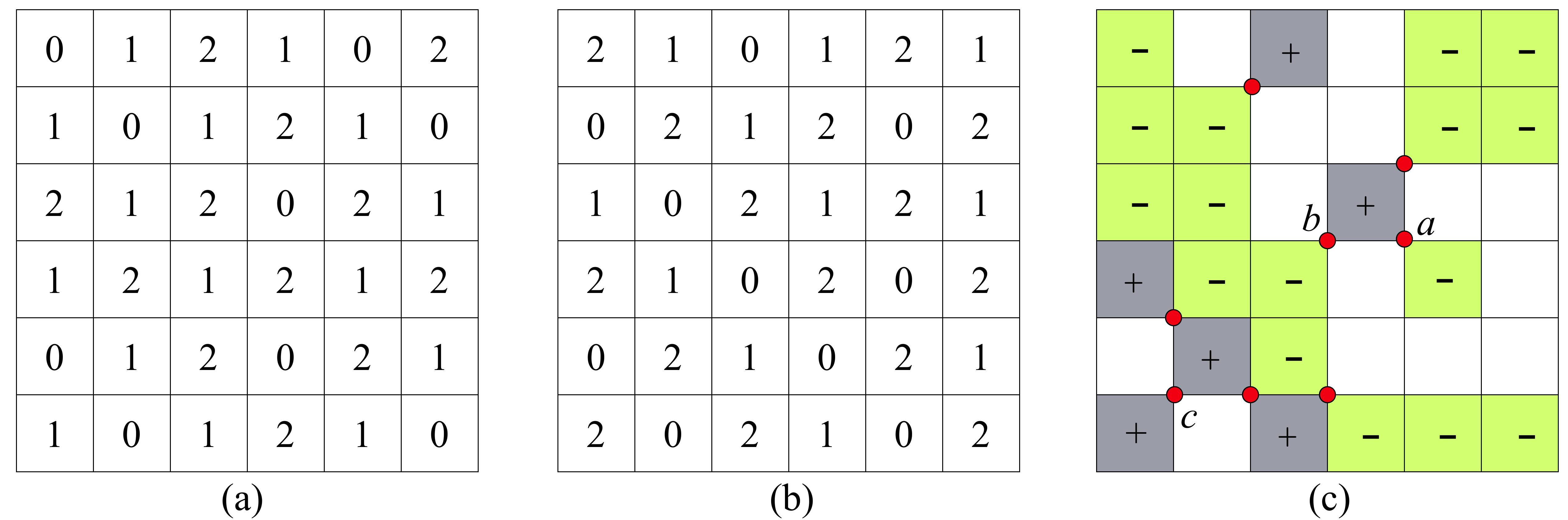}
\caption{3-colored grid graph $G$ (a,b) Colorings $\col_1$ and $\col_2$ (c) Difference graph $G(\col_1, \col_2)$.}
\label{fig:nonStandardUB-colors}
\end{figure}

Let $\col_1$ and $\col_2$ be two arbitrary $3$-colorings of $G$ (ignoring for now forced edges). We encapsulate 
the difference between $\col_1$ and $\col_2$ in a graph $G(\col_1, \col_2)$ with the same structure as $G$, which we refer to as the 
\emph{difference} graph.
This graph is obtained by partitioning $G$ 
into maximal polyominoes of three types: 
$(0)$-polyominoes include squares $s$ such that $\col_2(s) = \col_1(s)$; 
$(+1)$-polyominoes include squares $s$ such that $\col_2(s) = \col_1(s)+1$; and 
$(-1)$-polyominoes include squares $s$ such that $\col_2(s) = \col_1(s)-1$. 
A $t$-\emph{square}, with $t \in \{(0), (+1), (-1)\}$, is a square in the difference graph that belongs to a $t$-polyomino. 
\figref{nonStandardUB-colors}c shows the difference graph corresponding to the two colorings from Figures~\ref{fig:nonStandardUB-colors}a and~\ref{fig:nonStandardUB-colors}b. We view the square types $0$, $+1$ and $-1$ as colors and seek to partition the polyomino boundaries into a collection of uniform curves, meaning that the color difference across each arc of such a curve is uniform. 
These uniform curves will guide us into selecting a set of edges that forces $\col_1$ on $G$. 

In order to determine a partition of polyomino boundaries into uniform curves, we need to examine the way these boundaries intersect. 
Because a polyomino is composed of adjacent squares, the boundaries of two different polyominoes may not cross; they may, however, share square corners. Two adjacent polyominoes clearly share square corners, but it is also possible for three or four polyominoes to share a square corner (see square corners labeled $a$ and $b$ in \figref{nonStandardUB-colors}c, incident on three and four polyominoes, respectively).  Call a square corner where  three or four polyominoes meet a \emph{hub}; these are marked by thick red dots in \figref{nonStandardUB-colors}c. 

\begin{observation}
\label{obs:smooth}
At any hub $p$, two diagonally opposite squares incident on $p$ must be of the same type and the other two squares incident on $p$ may be of either the same or different types. 
\end{observation}

\begin{figure}[htbp]
\centering
\includegraphics[width=\linewidth]{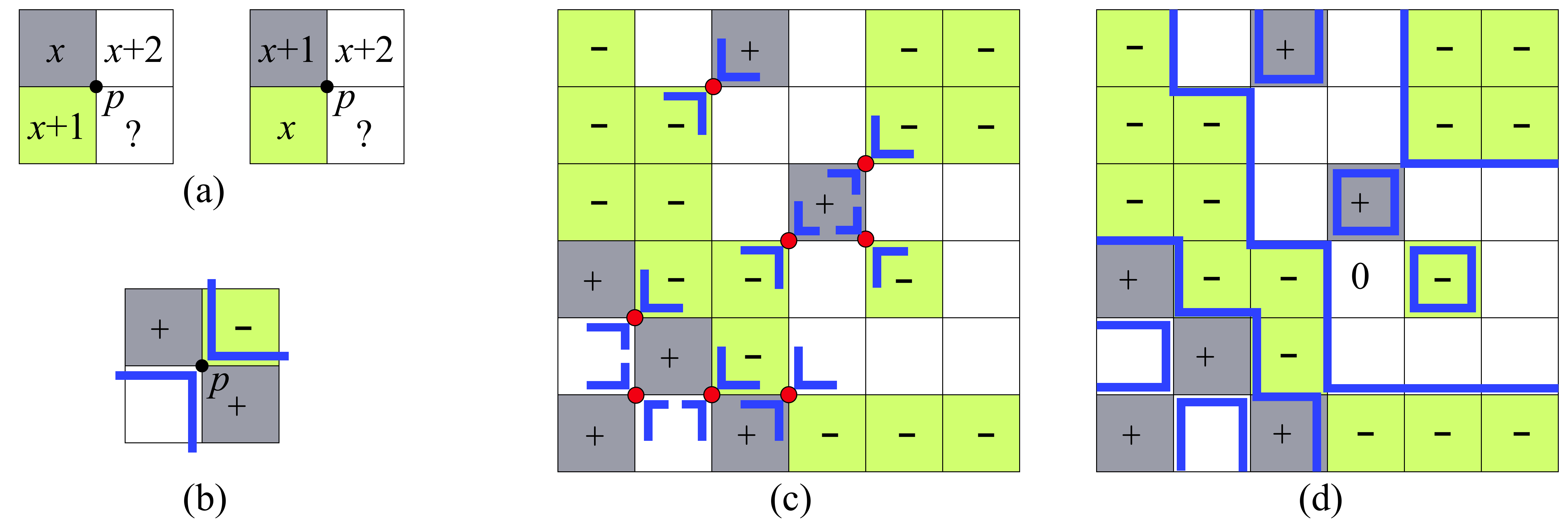}
\caption{(a) \autoref{obs:smooth}: two diagonally opposite squares must be of the same type  (b) Smoothing out $p$ (c) \figref{nonStandardUB-colors}c with all hubs (marked by thick dots) smoothed out 
(d) Uniform curves.}
\label{fig:nonStandardUB-paths}
\end{figure}

\noindent
To verify this observation, assume to the contrary that both pairs of diagonally opposite squares are of different types. Then two squares of the same type incident on $p$ are adjacent. For simplicity, assume that these are $0$-squares, so they have matching colors in $\col_1$ and $\col_2$ (the argument is symmetric for a different choice of square type).  
Then the other two squares incident on $p$ are also adjacent and therefore must be of different types (otherwise we would have only two polyominoes incident on $p$). 
Let $x$ and $x+1$ be the colors in $\col_1$ and $\col_2$ corresponding to the $(+1)$-square incident on $p$. Refer to \figref{nonStandardUB-paths}a. Then the colors in $\col_1$ and $\col_2$ corresponding to the $(-1)$-square must be $x+1$ and $x$ respectively, because any other pair of colors would result in adjacent identical colors, contradicting the fact that $\col_1$ and $\col_2$ are both valid colorings. This leaves $x+2$ as the only color available for the two adjacent $0$-squares incident on $p$, again in contradiction with the fact that  $\col_1$ and $\col_2$ are valid colorings.

\autoref{obs:smooth} is crucial in deciding how a uniform curve turns at a hub. It tells us that, if squares of three different types meet at a hub $p$, then two diagonally opposite squares are of different types, and the other two diagonally opposite squares are of the same type. 
In this case we smooth $p$ to two right angles corresponding to the squares of different types incident on $p$ (see \figref{nonStandardUB-paths}b). This leaves the two squares of the same type to the same side of each right angle, as required along each segment of a uniform curve. If the two squares of the same type incident to $p$ are in different polyominoes before smoothing $p$, then we view the two polyominoes as merging into a single one after smoothing $p$. 
If squares of only two types meet at a hub $p$, then any pair of diagonally opposite right angles can be used to smooth out $p$, because the other two squares are guaranteed to be of the same type. \figref{nonStandardUB-paths}c shows the result of smoothing all the hubs in the difference graph from \figref{nonStandardUB-colors}c.

Once all hubs have been eliminated by this smoothing process, the union of all polyomino boundaries contains only pairwise disjoint components. These components are precisely the uniform curves in the graph formed from the difference between $\col_1$ and $\col_2$. (See \figref{nonStandardUB-paths}d.) 

\begin{observation}
\label{obs:o1}
For each uniform curve $\C$ in $G(\col_1, \col_2)$, the difference$\pmod 3$ between two colors across an arc of $\C$ is different in $\col_1$ and $\col_2$. 
\end{observation}
To see this, pick an arbitrary uniform curve $\C$ and let $P_1$ and $P_2$ be the two polyominoes adjacent to either side of $\C$ in the graph formed by the difference between $\col_1$ and $\col_2$. Pick an arbitrary arc $e$ of $\C$, and let $v_1 \in P_1$ and $v_2 \in P_2$ be the two adjacent vertices on either side of $e$. We now consider all possible combinations of polyomino types. If $P_1$ is a $(0)$-polyomino, then $P_2$ is either a $(+1)$ or a $(-1)$-polyomino. In the first case we have $\col_2(v_2) = \col_1(v_2) + 1$, and because $\col_2(v_1) = \col_1(v_1)$, we get $\col_2(v_2) - \col_2(v_1) = \col_1(v_2) - \col_1(v_1) + 1$. In the latter case, the argument is symmetric with $+1$ replaced by $-1$. 
Finally, if $P_1$ is a $(+1)$-polyomino and $P_2$ is a $(-1)$-polyomino, we have $\col_2(v_1) = \col_1(v_1) + 1$, $\col_2(v_2) = \col_1(v_2) -1$, which together yield $\col_2(v_2) - \col_2(v_1) = \col_1(v_2) - \col_1(v_1) - 2$. This settles \autoref{obs:o1}. 

\begin{lemma}
\label{lem:fs2}
Let $F$ be a set that includes, for every uniform curve $\C$ in $G(\col_1, \col_2)$, an edge across $\C$ with its end vertices colored according to $\col_1$. Then $\col_2$ does not agree with the coloring imposed by $F$. 
\end{lemma}
\begin{proof}
Let $\C$ be an arbitrary uniform curve in the difference graph $G(\col_1, \col_2)$.  By the lemma statement, $F$ includes an edge $(x, y)$ across $\C$ such that the difference between the colors of $x$ and $y$ equals $\col_1(y) - \col_1(x)$. 
By \autoref{obs:o1}, $\col_2(y) - \col_2(x) \neq \col_1(y) - \col_1(x)$, therefore $\col_2$ does not respect $F$. 
\end{proof}

\paragraph{The Algorithm. }
Next we describe an $O(m^2n^2)$ time algorithm that determines a minimum forcing set of a given $3$-colored $m \times n$ grid graph $G$. Let $\col$ denote the coloring function for $G$. 

\begin{figure}[htbp]
\centering
\includegraphics[width=0.5\linewidth]{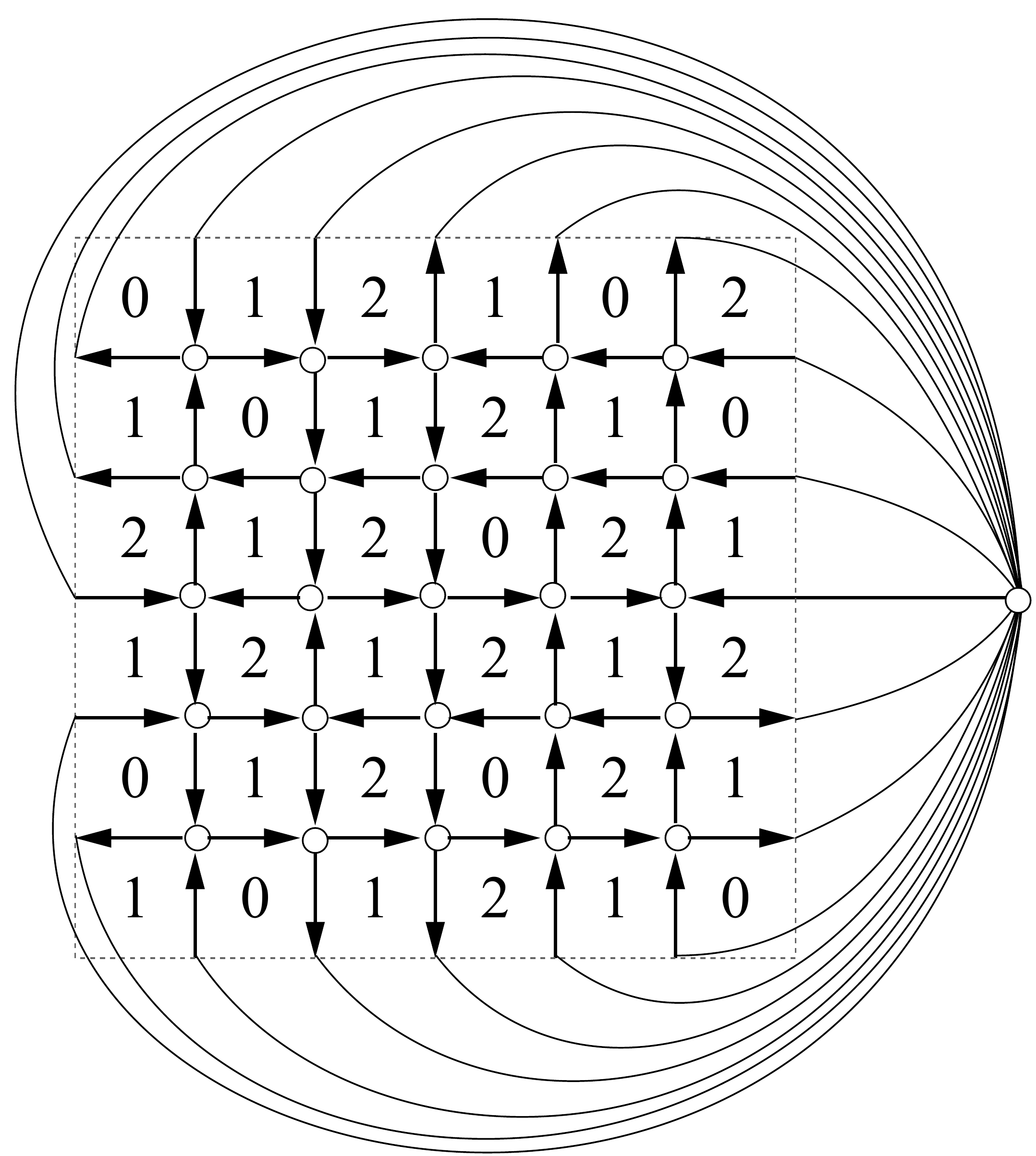}
\caption{Planar directed graph $H$ encapsulating all uniform curves.}
\label{fig:nonStandardUB-graph}
\end{figure}

The algorithm consists of two main steps. In a first step we derive from $G$ a planar directed graph $H$ that encapsulates all uniform curves of $G$. In a second step we select from $H$ a minimum set of edges whose removal leaves $H$ acyclic, and show that this is a minimum forcing set for $G$. \autoref{alg:mfs} outlines the key steps of this algorithm. 

\vspace{1mm}
\SetKwBlock{Stageone}{{\bf 1. Compute the directed graph $H$:}}{}%
\SetKwBlock{Stagetwo}{{\bf 2. Compute a forcing set $F$:}}{}%
\begin{algorithm}
\centerline{{\sc MinimumForcingSet}(Grid graph $G$, Coloring $\col$)}
{\hrule width 0.92\linewidth}\vspace{0.8em}
\Stageone{ 
Set $V(H) \leftarrow \{\mbox{all interior square corners}\}$. \\
Add the outer node to $V(H)$.\\
Set $E(H) \leftarrow \{\mbox{all interior arcs}\}$. \\
(* Arcs incident to the outer boundary extend to meet the outer node.*) \\
\ForEach {arc $e \in E(H)$} {
Let $s_i$ and $s_j$ be the two squares adjacent along $e$. \\
\lIf{$\col(s_j) = \col(s_i)+1$} \\ 
\Indp {orient $e$ clockwise around $s_i$}  \\
\Indm \lElse{orient $e$ clockwise around $s_j$}. 
}
}
\Stagetwo{
Initialize $F \leftarrow \emptyset$. 
Find a minimum feedback arc set $J$ in $H$~\cite{L76Thesis, Ram88}. \\ 
\ForEach {arc $e \in J$} {
Add to $F$ the edge crossing $e$. 
}
\Return $F$.
}
\noindent{\hrule width 0.92\linewidth} \vspace{1mm} 
\caption{Minimum Forcing Set algorithm.}
\label{alg:mfs}
\end{algorithm}

We now turn to describing the first step in detail. The graph $H$ has one node for each interior square corner of $G$, and one node for the whole outer boundary of $G$. Refer to \figref{nonStandardUB-graph}. The edges of $H$ are the interior 
arcs, with those incident to the outer boundary extended to meet the representative node.

For every pair of squares $s_i$ and $s_j$ adjacent along an arc $e$, we orient $e$ clockwise around $s_i$ if $\col(s_j) = \col(s_i)+1\pmod{3}$, and counterclockwise otherwise. (Note that swapping the roles of $s_i$ and $s_j$ does not change this rule for directing arcs.) The construction of $H$ is linear in the size of $G$, which is $O(mn)$. 

\begin{observation}
\label{obs:uniform}
Every directed cycle of $H$ is a uniform curve. 
\end{observation}
This observation follows immediately from the fact that all arcs have the same orientation along a directed cycle $\C$ in $H$. This means that the difference (mod 3) between the two colors across every arc of $\C$ must be the same (either $+1$ across all arcs, or $-1$ across all arcs), otherwise two arcs along $\C$ would have opposite orientations. 

\begin{lemma}
\label{lem:curves}
For any coloring $\col' \neq \col$, the set of directed cycles in $H$ includes all uniform curves in the difference graph $G(\col, \col')$.
\end{lemma}
\begin{proof}
Pick an arbitrary uniform curve $\C$ in $G(\col, \col')$. By definition, the color difference$\pmod{3}$ is the same across every arc of $\C$. This implies that all arcs of $\C$ are oriented the same way in $G(\col, \col')$ and therefore $\C$ is a directed cycle in $H$. 
\end{proof}

\begin{figure}[htbp]
\centering
\includegraphics[width=\linewidth]{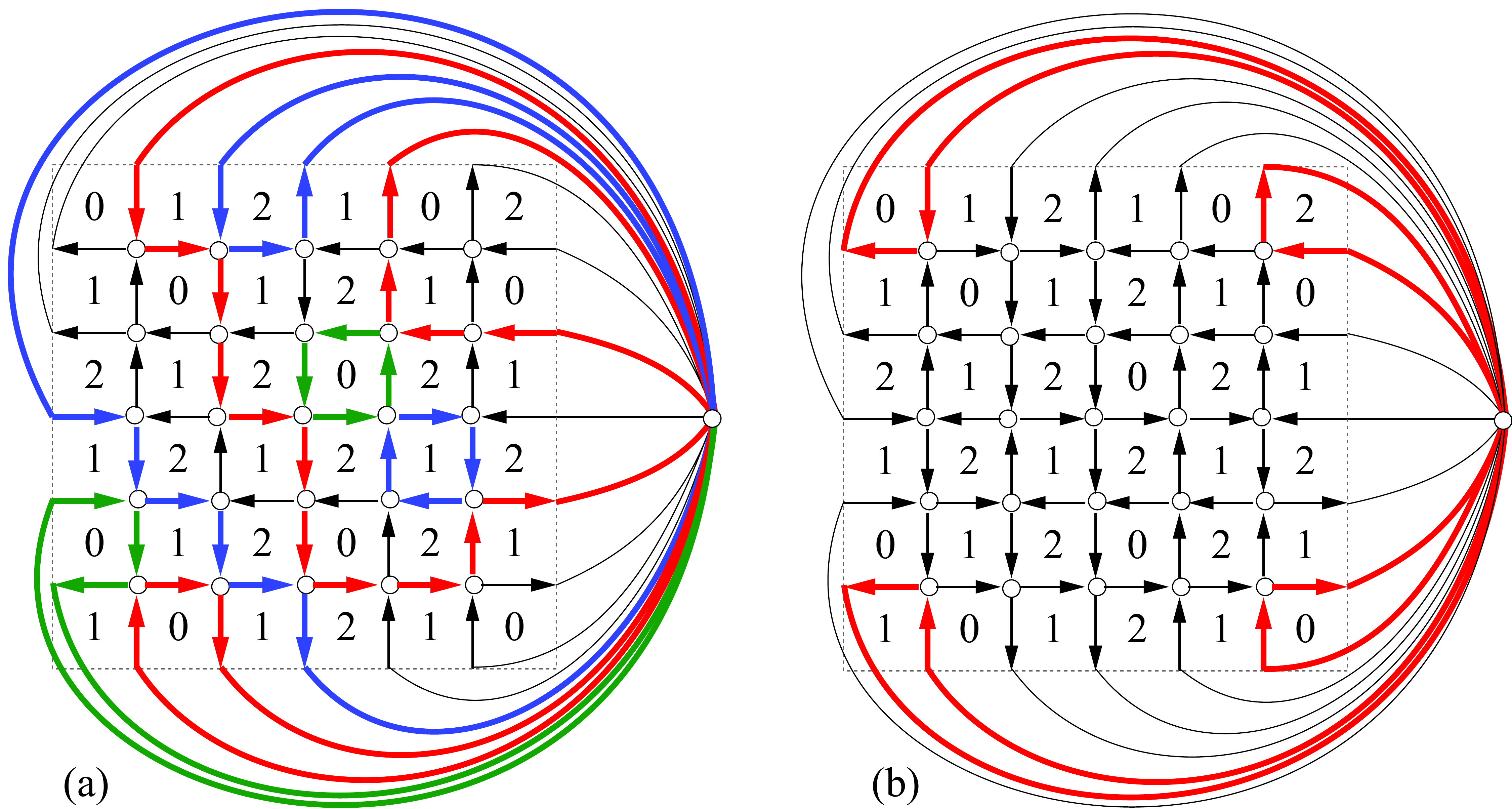}
\caption{Uniform curves in the difference graph from~ \figref{nonStandardUB-paths}c.} 
\label{fig:nonStandardUB-graph2}
\end{figure}
\noindent
\figref{nonStandardUB-graph2}a outlines the uniform curves in the difference graph from \figref{nonStandardUB-paths}c. 
\figref{nonStandardUB-graph2}b outlines four other uniform curves that potentially correspond to different valid colorings of $G$. 
\begin{theorem}
\label{thm:main}
Let $F$ be a set of edges. Coloring $\col$ is forced by $F$ if and only if $F$ includes an edge across every directed cycle of $H$.
\end{theorem}
\begin{proof}
To prove this theorem in the `if' direction, assume that the coloring $\col$ is forced by $F$, but $F$ does not include an edge in every directed cycle of $H$. Let $\C$ be such a directed cycle. By \autoref{obs:uniform}, $\C$ is a uniform curve. By \autoref{cor:uniform}, every forcing set ($F$ in particular) must include an edge across $\C$, a contradiction. 
To prove the theorem in the `only if' direction, assume that $F$ includes an edge across every directed cycle of $H$, but $F$ does not force $\col$. This implies that there exists another valid coloring $\col'$ that respects $F$. By \autoref{lem:curves}, each uniform curve in the difference graph $G(\col, \col')$ is a directed cycle in $H$. By the theorem statement $F$ includes an edge across each such directed cycle. This along with \autoref{lem:fs2} implies that $\col'$ does not respect $F$, a contradiction. 
\end{proof}

Theorem~\ref{thm:main} suggests that finding a minimum forcing set $F$ for $G$ can be reduced to finding a minimum feedback arc set in $H$ (that is, a minimum set of arcs whose removal leaves $H$ acyclic). This problem is known to be polynomial 
for the class of reducible flow graphs, which include planar digraphs~\cite{L76Thesis, Ram88}.
More precisely, a minimum feedback arc set in the planar digraph $H$ can be determined  in time quadratic in the number of arcs in $H$~\cite{Ram88}, which is $O(m^2n^2)$.
We use the solution to the minimum feedback arc set problem on $H$ to extract a minimum forcing set for $G$: corresponding to each arc $e$ in the minimum feedback arc set for $H$, we include in $F$ the edge crossing $e$.  By Theorem~\ref{thm:main}, $F$ is a forcing set for $G$. The running time of the algorithm that determines $F$ is $O(m^2n^2)$. 

\subsection{Lower Bounds for Non-Standard Miura-ori}
\label{sec:lowerboundm+n+2}
In this section we show the existence of $m \times n$ Miura-ori
folding patterns whose forcing sets are of size exactly $m + n - 2$, the smallest possible for any Miura-ori pattern. 

\begin{theorem}[Lower Bound]
There exist $m \times n$ grid graphs whose forcing sets require exactly $m+n-2$ edges. Here $m$, $n$ refer to the number of vertices in a column or row, respectively, of the grid graph. \0{inserted ``a column or row, respectively, of ``}
\end{theorem}
\begin{proof}
Consider a grid graph with diagonal stripes of one of the three colors in cyclic order, as depicted in \figref{nonStandardLB}a. 
Each boundary-to-boundary path running alongside a diagonal stripe of squares separates some color $x$ on one side from an adjacent color $y$ on its other side.
%
\begin{figure}[bhtp]
\centering
\includegraphics[width=0.7\linewidth]{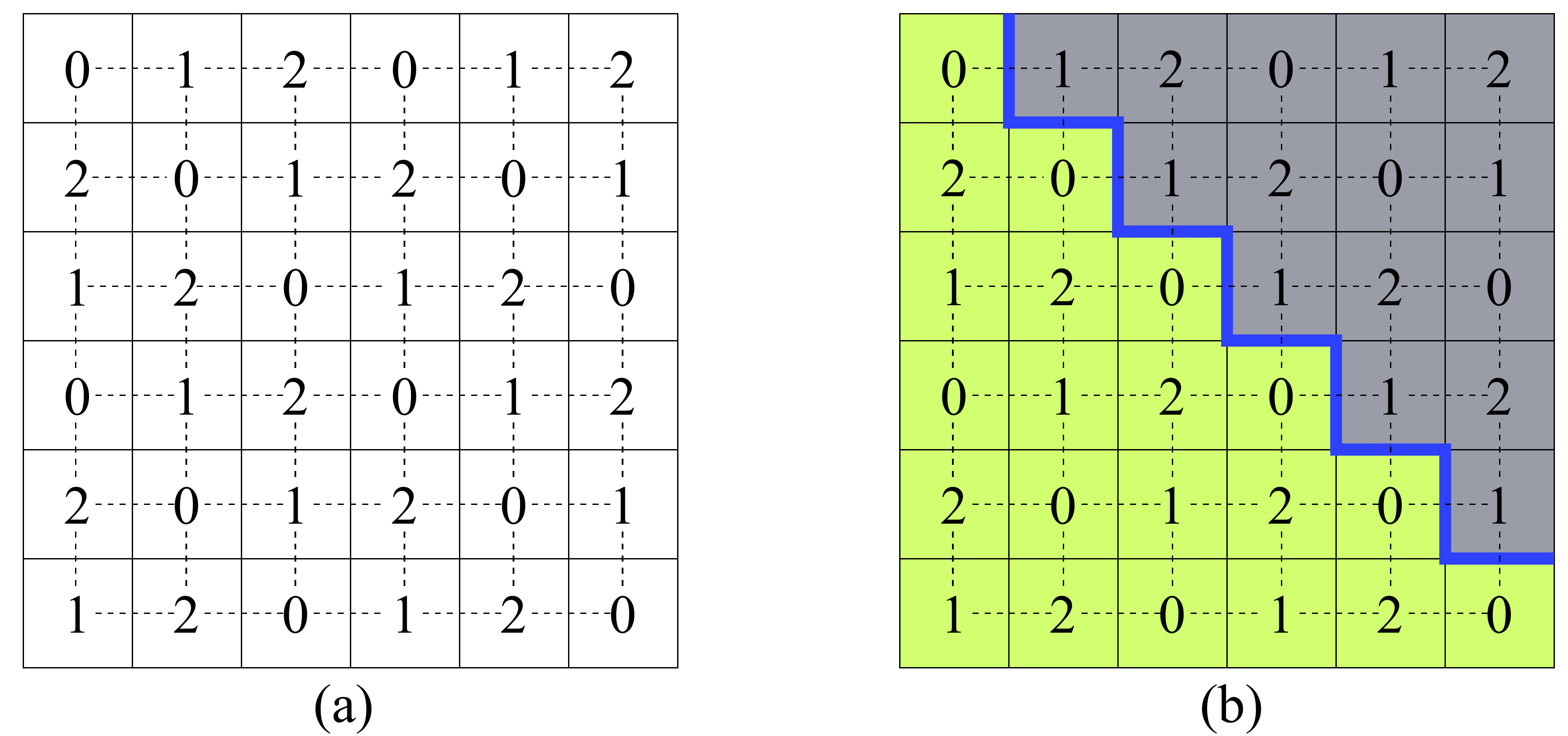}
\caption{(a) 3-colored grid graph (b) Diagonal boundary-to-boundary path is a uniform curve.}
\label{fig:nonStandardLB}
\end{figure}
%
(See path marked in \figref{nonStandardLB}b, which separates colors $0$ and $1$.)  
Such a path is a uniform curve, and by \autoref{cor:uniform} it must be crossed by an edge in the forcing set. 
The number of such uniform curves is exactly $m+n-2$, because each crosses exactly two of the $2(m-1+n-1)$ boundary edges. Because these diagonal paths (and therefore the sets of edges they cross) are pairwise disjoint, the forcing set must include at least as many edges as the number of uniform curves, which is $m + n-2$. 
\end{proof}

This bound is tight, as for any Miura-ori folding pattern, every boundary arc must be part of one of the uniform curves, each uniform curve can only include two such arcs, and each curve must be crossed by a forcing set. Therefore, there can be no pattern requiring fewer than $m+n-2$ edges in its smallest forcing set.

\subsection{Upper Bounds for Non-Standard Miura-ori}

In this section we give a simple $O(mn)$ time algorithm that establishes an upper bound of
$\lceil \frac{mn}{2} \rceil$ on the size of a forcing set
for any flat-foldable MV-assignment, 
making the standard Miura-ori in a sense the worse case. \0{For non-standard Miura, this forcing set may not be of minumum size; finding forcing sets of minimum size
for non-standard Miura-ori is addressed next in \autoref{sec:standardmiura}.} In this section we again make use of the grid graph representation of the Miura-ori discussed in \autoref{sec:gridgraph}. 

\begin{theorem}
\label{thm:miura_Ndiv4}
Given any locally flat-foldable MV assignment of a $m \times n$ Miura-ori crease pattern, there is an algorithm that finds a forcing set of size $\lceil \frac{mn}{2} \rceil$ in time $O(mn)$.
\end{theorem}
\begin{proof}
Let $G$ be the $m \times n$ $3$-colored grid graph corresponding to the
Miura-ori MV assignment with the color of the upper-left vertex fixed to 0. 
We create a forcing set $F$ for $G$ by assigning an orientation to the edges in $G$ two rows at a time from top to bottom, and within each pair of rows we work left to right.
When done, the set $F$ will be a set of directed edges. Each directed edge $(v_i,v_j)$ with source $v_i$ and destination $v_j$
has weight $w(v_i,v_j)$ of $\pm 1$
such that $\col(v_j) = \col(v_i) + w(v_i,v_j) \mod 3$, where $\col(v)$ is the color of vertex $v$.     

We will assume that $m$ (the number of rows) is even; the algorithm is the same when $m$ is odd, with
just a minor modification (not discussed  here) to handle the last row.
Starting with the top two rows, 
we begin by adding to $F$ the vertical edge directed from the upper-left vertex to the one below. 
Because the color of the upper-left vertex is fixed to 0, 
this forces the color of the vertex below it.  
We then look at each successive $2 \times 2$ block formed by two already-forced vertices $v_\ell$, $v_\ell'$ and
the two new vertices $v_r,v_r'$ to the right of them. 
If it is a $2$-colored block, then adding to $F$ the directed edge $(v_r,v_r')$ with its associated weight 
is enough to force the colors of the new vertices. (See \figref{miura_Ndiv4}a.)
It is easy to verify that the only coloring for the new vertices $v_r,v_r'$
that is consistent with the added forcing edge and the already forced colors of vertices $v_\ell,v_\ell'$ is
one in which $\col(v_r )= \col(v_\ell')$ and $\col(v_r') = \col(v_\ell)$.
If the block is $3$-colored, then one of the two new vertices has neighbors of both colors within the block. (See vertex $v_r'$ in \figref{miura_Ndiv4}b.)  
In this case, adding to $F$ the horizontal edge connecting an already forced vertex to the new vertex that does not have neighbors of both colors is enough to
force the colors of the new vertices.
\begin{figure}[t]
\centering
\includegraphics[width=0.7\linewidth]{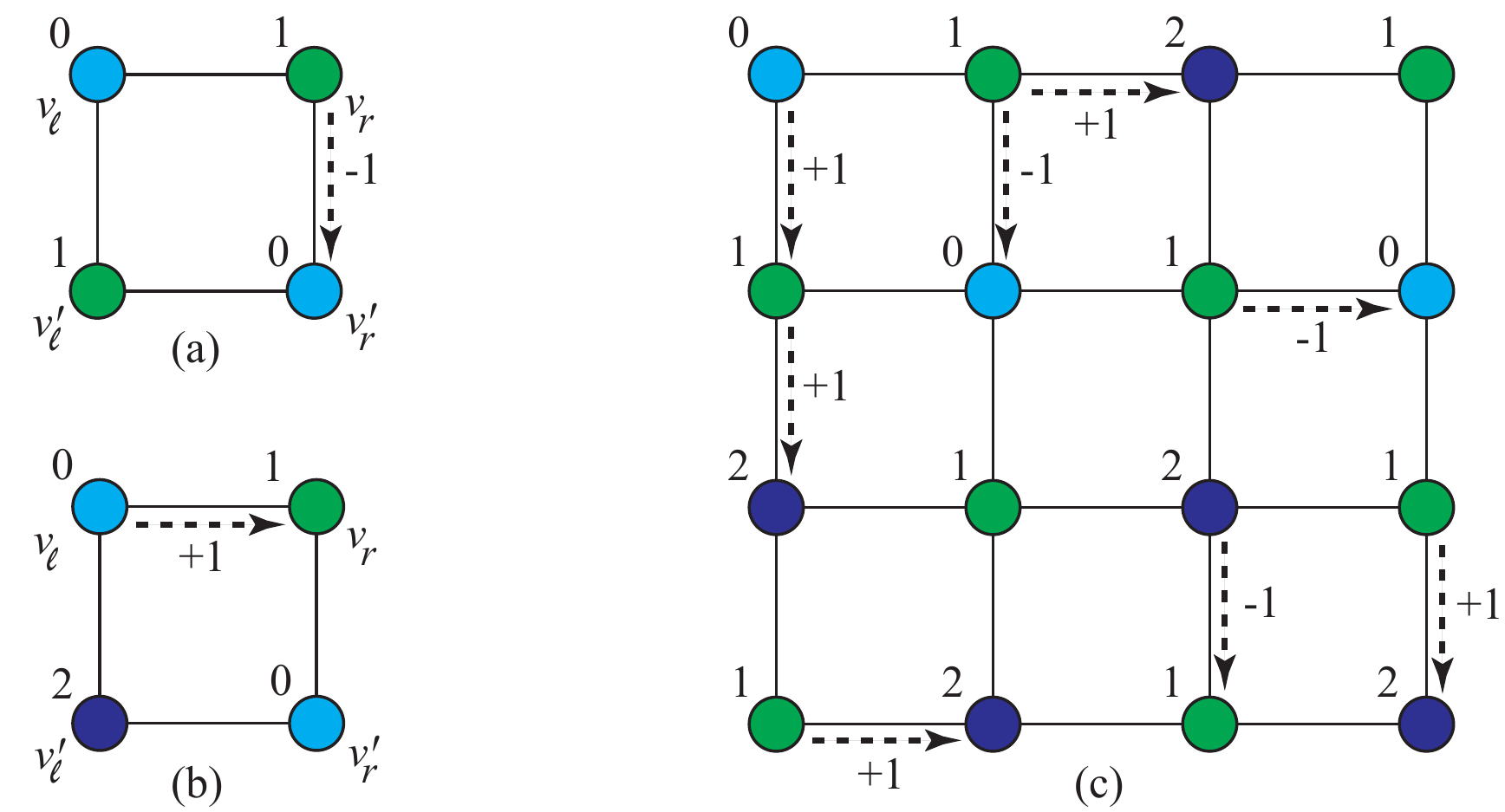}
\caption{
(a) $2$-colored block rule. (b) $3$-colored block rule. (c) Forcing set generated by the algorithm described in the proof in Theorem~\ref{thm:miura_Ndiv4}. 
}
\label{fig:miura_Ndiv4}
\end{figure}

After forcing all colors on a pair of rows in this manner, we transition to the next row pair below it by
looking at the $2 \times 2$ block formed by the first two vertices on the bottom row of the previous row pair and
the first two vertices on the 
top row of the new row pair. Because the top two vertices in the block are already forced, we can apply the same rules from above to add one
edge to $F$ that forces the two lower vertices in the block. (Note that the correctness of the block rules does not depend on which pair of
adjacent vertices in the block are already forced.) 
We do this again with the first $2 \times 2$ block of the new row pair, and from there continue left to right across the new row pair as before.
This repeats until all row pairs have been colored, so the algorithm completes in time linear in the size of $G$, which is $O(mn)$. 

When processing each $2 \times 2$ block, we gain two new colored grid vertices for only one forced edge. Therefore the total number of forced edges is $\frac{mn}{2}$. 
(When $m$ and $n$ are both odd, we need the ceiling of this term because one edge is used to force the bottom right vertex.)
With the established correspondence between grid graph colorings and Miura-ori crease pattern MV assignments, the 
forcing set $F$ of $G$ immediately gives a forcing set for the Miura-ori MV assignment where each forcing edge in $F$ is replaced by its dual crease. 
\end{proof}

\section{Forcing Set Characterization}
\label{sec:forcingsetcharacterization}
In this section we describe algorithms for testing whether a given set is forcing for a given Miura-ori crease pattern and whether a given partial mountain-valley assignment can be completed to a locally flat-foldable Miura-ori. 
Combining these two results together will give us an efficient algorithm for the problem of, given a forcing set, determining the folding pattern that it forces.

\subsection{From Grid Colorings to Eulerian Digraphs}

The following lemma characterizes the locally flat foldable MV assignments (or equivalently 3-vertex-colorings of a grid) in terms of the properties of the digraph used in our algorithm for finding a minimum forcing set.
The lemma will  be used shortly in testing for forced local flat foldability (\autoref{sec:testfold}).

\begin{lemma}
\label{lem:euler}
Let $G$ be a grid graph. Then the construction of an Eulerian digraph $H$ from a 3-coloring of $G$, in Step 1 of \autoref{alg:mfs}, is a bijection between 3-colorings of $G$ (with a fixed choice of color for one vertex) and Eulerian orientations of the graph formed by adding one outer vertex to~$G$.
\end{lemma}
\begin{proof}
Consider a $2\times2$ block of squares and let $x$ be the color of the top left square $s$ in this block. Assume first that $G$ has a valid $3$-coloring. The two squares adjacent to $s$ in this block can be of either the same or different colors. If they are of different colors, then the square diagonally opposite to $s$ must have color $x$ and the arc orientation is as shown in \figref{nonStandardUB-Eulerian}a. If the two squares adjacent to $s$ have the same color (as in Figures~\ref{fig:nonStandardUB-Eulerian}b--\ref{fig:nonStandardUB-Eulerian}e), then there are two choices for the  color of the square diagonally opposite to $s$. All possible cases are depicted in \figref{nonStandardUB-Eulerian}, and each of these cases yields the indegree and outdegree equal to 2 at the center node. This shows that $H$ is Eulerian. 

In the other direction, if $H$ is Eulerian, then we can process $2\times2$ block of squares, starting with the top left block, and assign colors to the block squares according to the patterns showed in \figref{nonStandardUB-Eulerian}. We then look at each successive $2\times2$ block formed by the right column of the previous block and the one to its right, so that each successive block includes two colored squares to guide the coloring of the other two squares. After coloring all squares in the first row pair, we transition to the next row pair formed by the bottom row of the previous pair and the one immediately below it, and continue this process. 
The final result is a valid $3$-coloring of $G$.  
\end{proof}

\begin{figure}[hpt]
\centering
\includegraphics[width=\linewidth]{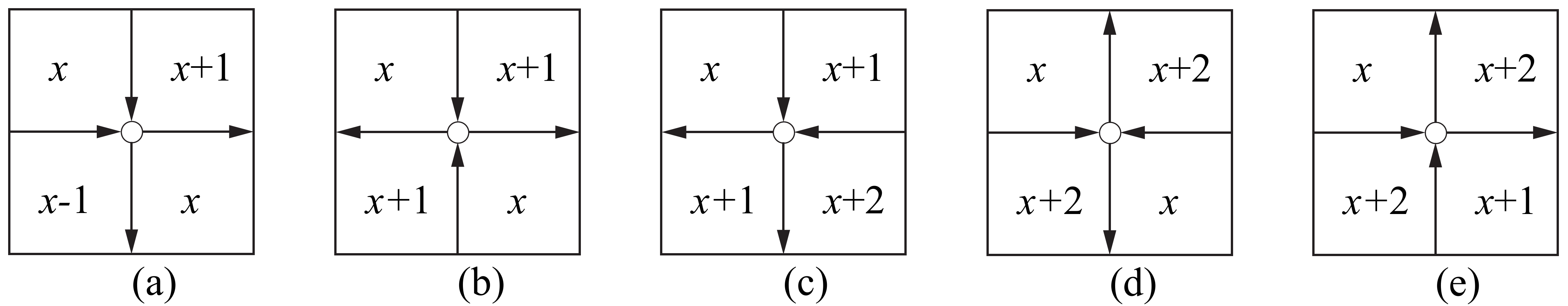}
\caption{The in degree and the out degree are both $2$.}
\label{fig:nonStandardUB-Eulerian}
\end{figure}

\subsection{Testing Whether a Set is Forcing}
Given a $3$-colored $m\times n$ grid graph $G$ and a set $F$ of edges in $G$, we can easily determine in time $O(mn)$ whether $F$ is a forcing set or not. To do so, we first derive from $G$ a directed graph $H$ as described in Step 1 of the minimum forcing set algorithm (see \autoref{alg:mfs}). As mentioned earlier, this step takes time $O(mn)$. 
Corresponding to each edge $f \in F$, we delete from $H$ the dual arc crossing $f$. Let $H'$ be the graph obtained from $H$ after removing all such arcs. Theorem~\ref{thm:main} implies that  $F$ is a forcing set if and only if $H'$ is acyclic. Thus testing the forcing property for $F$ reduces to testing the acyclic property for $H'$, which
can be accomplished by means of a simple depth-first traversal of $H$ in time $O(mn)$.

\subsection{Testing for Local Flat Foldability}
\label{sec:testfold}
Given an assignment of mountain/valley to a subset $S$ of the creases in a $m \times n$ Miura-ori pattern, our goal in this 
section is to test whether there exists a locally flat-foldable fold on all pattern creases that is consistent 
with the input $S$. We show that this task can be accomplished in $O(mn\sqrt{mn})$ time.

Let $G$ be the graph representation of the Miura pattern, with vertices interior to the squares of a $m \times n$ grid (as in \figref{uniform}a). 
 Precolor $0$ the top left vertex of $G$ and assign weights $+1$ or $-1$ to all edges of $G$ associated with creases in $S$, consistent with the MV assignment of the corresponding creases. Each edge not associated with a crease in $S$ has a null weight. Based on the correspondence between local flat foldability of the Miura-ori crease pattern and $3$-colorability of $G$ (established in \autoref{sec:gridgraph}), testing for local flat-foldability reduces to showing that $G$ has a valid $3$-coloring that respects the non-null edge weights. 

\begin{figure}[hpt]
\centering
\includegraphics[width=\linewidth]{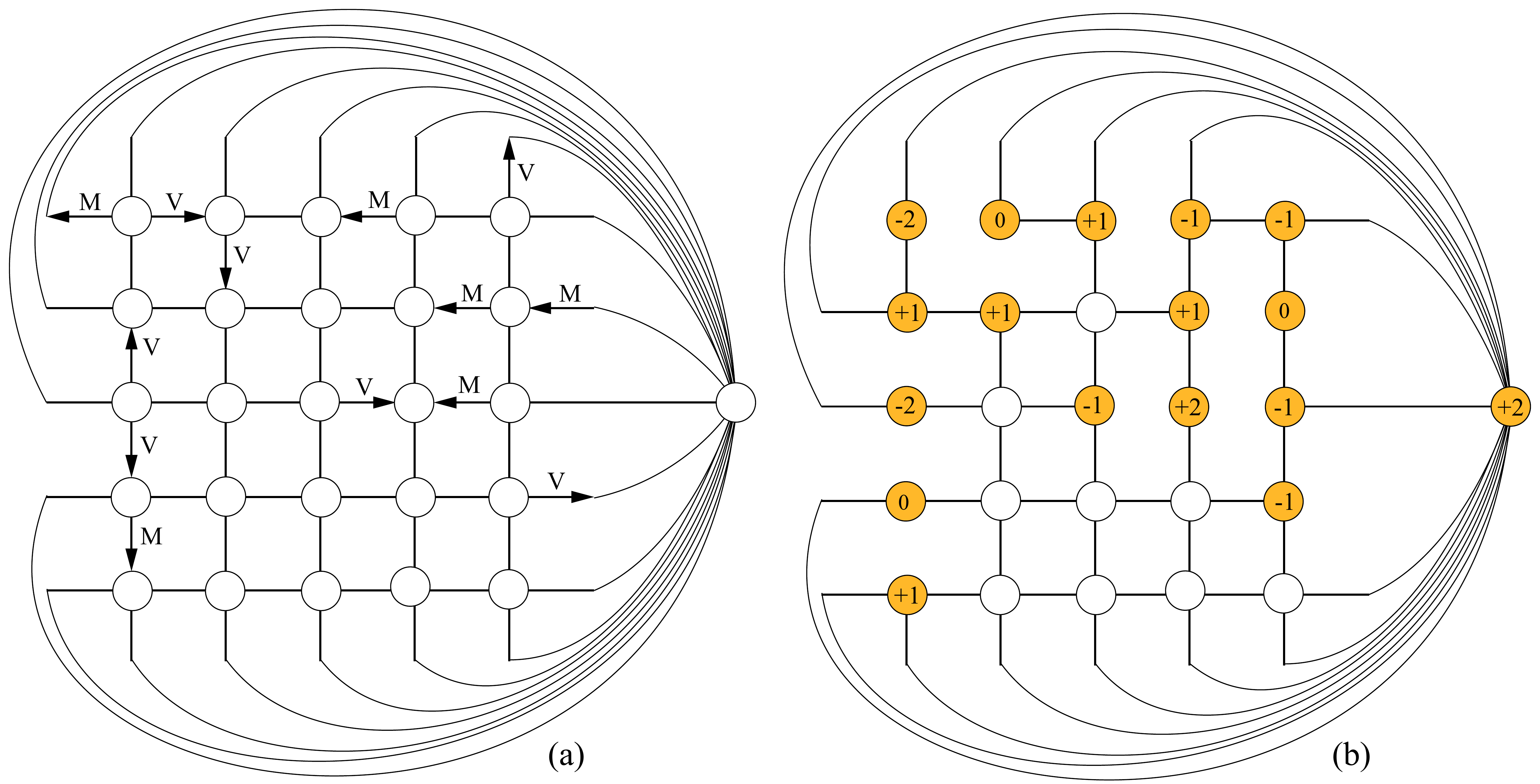}
\caption{(a) Graph representation $H$ (b) Flow network setup with demands ($-$) and supplies ($+$).}
\label{fig:nonStandard-testing}
\end{figure}

Let $H$ be the directed representation for $G$ derived as in Step 1 of the the minimum forcing set algorithm (see \autoref{alg:mfs}). 
Note that 
arcs in $H$ corresponding to creases in $S$ 
will be oriented in $H$, but the rest of them start as undirected arcs. See \figref{nonStandard-testing}a for an example. 
By \autoref{lem:euler}, $G$ has a valid $3$-coloring if and only if $H$ can be extended to a directed Eulerian graph in which  every  node (with the exception of the outer node) has  in degree and out degree equal to $2$. Thus our goal is to orient all undirected  arcs  of $H$ such that the resulted graph is directed Eulerian. 

To determine an Eulerian orientation for $H$, we set up a network flow on the undirected arcs of $H$, where each arc has capacity $+1$ in both directions. Corresponding to each arc $(u,v)$ directed from $u$ to $v$ in $H$, we add one unit of demand at the tail $u$ and one unit of supply at the head $v$. Thus the total supply equals the total demand. 
\figref{nonStandard-testing}b shows the flow network corresponding to the directed graph from \figref{nonStandard-testing}a; negative values at nodes indicate demand and positive values indicate supply. 
We need to determine a feasible flow from supply nodes to demand nodes that satisfies the demands at all nodes, subject to the unit capacity constraints. Because there are no sources or sinks involved here, determining a feasible flow in $H$ reduces to finding a circulation in $H$, which can be accomplished in time $O(mn\sqrt{mn})$ using Miller's algorithm (see section 4 in~\cite{Miller95}). 
If there is no solution to the flow problem, then 
there is no locally flat-foldable fold on all pattern creases that is consistent with the input. Otherwise, the solution to the flow problem is an Eulerian orientation that includes all the input creases. \figref{nonStandard-testing2}a shows an Eulerian orientation that satisfies the supply and demand requirements for the flow network depicted in \figref{nonStandard-testing}b. 

\begin{figure}[hpt]
\centering
\includegraphics[width=\linewidth]{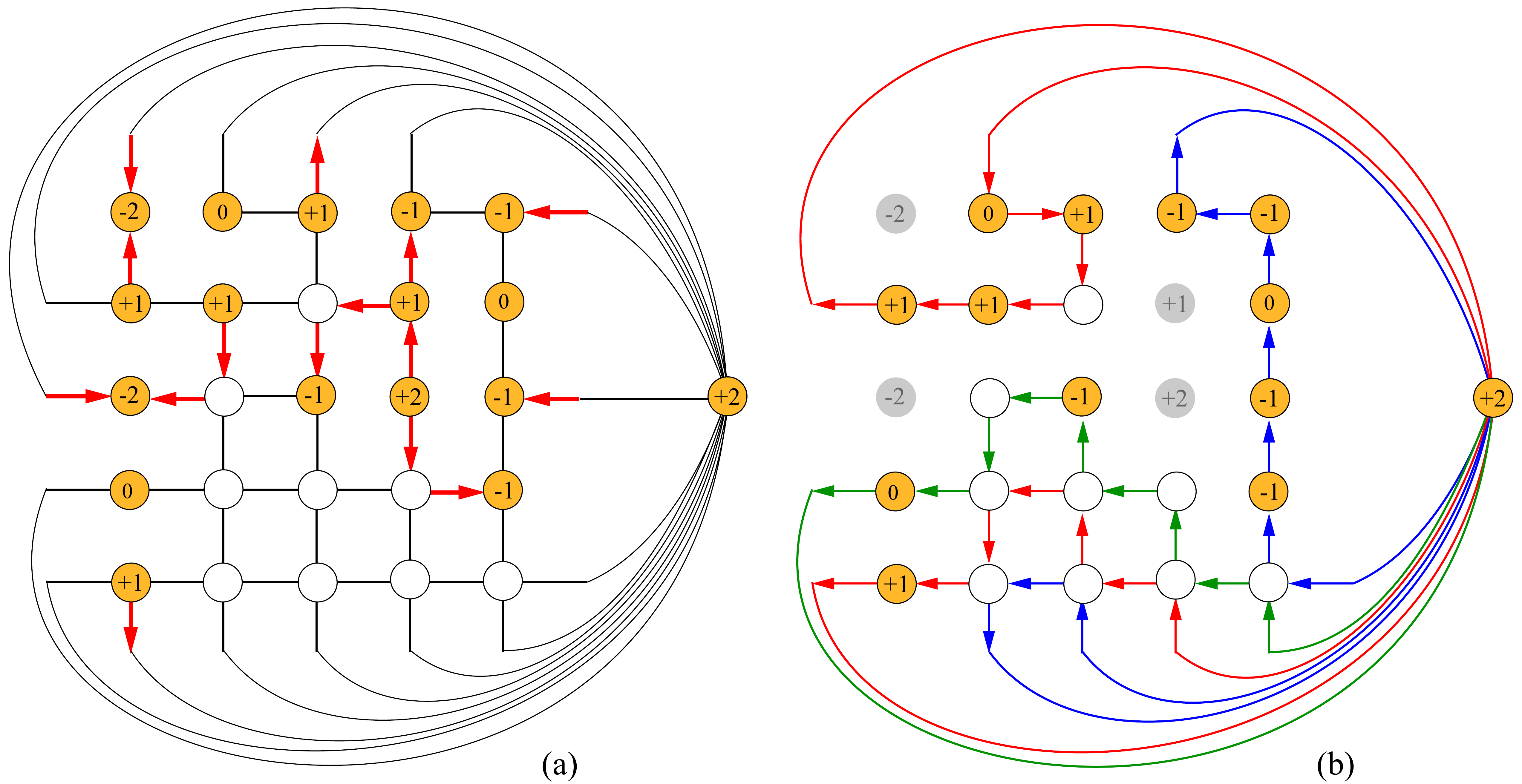}
\caption{(a) Network flow solution (thick directed arcs). (b) Partitioning $H'$ into cycles.}
\label{fig:nonStandard-testing2}
\end{figure}

The Eulerian orientation determined in the previous step includes all creases in the input set $S$, but it may leave some arcs in  $H$ undirected. Let $H'$ be the subgraph of $H$ consisting of all undirected arcs (it may have several connected components). Each node in $H'$ has even degree, therefore $H'$ is an Eulerian undirected graph. By Veblen's theorem~\cite{Veb-AM-12}, $H'$ can be partitioned into simple disjoint cycles. This can be accomplished in time linear in the number of graph arcs, which is $O(mn)$. 
Choose an arbitrary orientation for each cycle. (See \figref{nonStandard-testing2}b for an example.) The result is the desired Eulerian orientation on the whole graph $H$. (See \figref{nonStandard-testing3}a,  which combines the arc orientations from Figures~\ref{fig:nonStandard-testing}a,~\ref{fig:nonStandard-testing2}a and~\ref{fig:nonStandard-testing2}b.)
%
\begin{figure}[hpt]
\centering
\includegraphics[width=\linewidth]{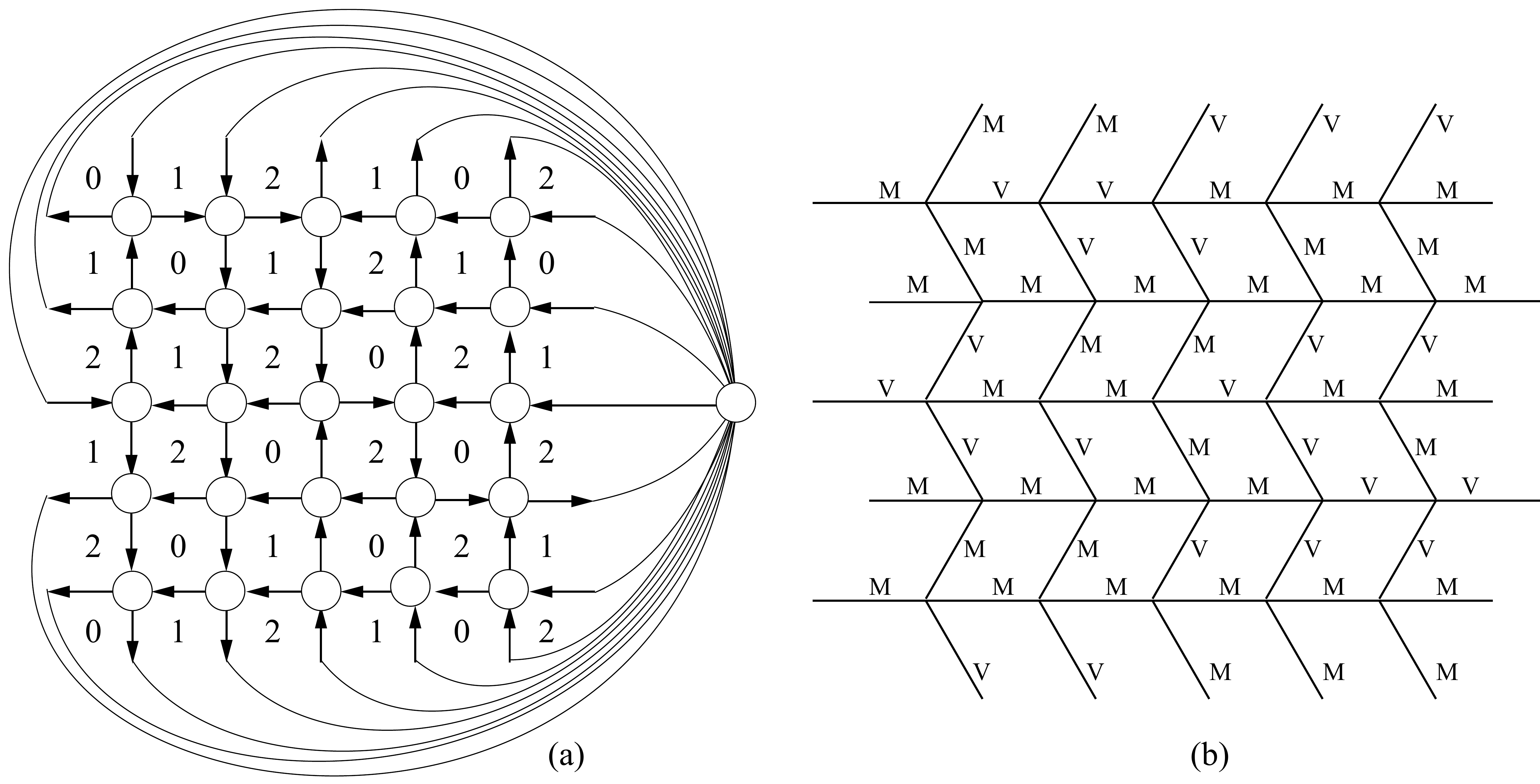}
\caption{(a) $3$-coloring consistent with the Eulerian orientation (b) Equivalent MV assignment.}
\label{fig:nonStandard-testing3}
\end{figure}

From the Eulerian orientation of $H$ we determine a $3$-coloring of the grid graph following the approach used in Step 1 of the minimum forcing set algorithm (\autoref{alg:mfs}): if an arc $e$ shared by squares $s_i$ and $s_j$ is oriented clockwise around $s_i$, then $\col_(s_j) = \col(s_i)+1$, otherwise $\col_(s_i) = \col(s_j)+1$. \figref{nonStandard-testing3}a shows the correspondence between the $3$-coloring and the Eulerian orientation of $H$.  
Finally we use the established correspondence between a $3$-coloring of the grid graph and a folding pattern to determine a mountain/valley assignment for the all creases, consistent with the input $S$. \figref{nonStandard-testing3}b shows the mountain/valley assignment corresponding to the $3$-coloring from \figref{nonStandard-testing3}a.

\section{Tight Controlling Sets}
\label{sec:tightcontrol}
This section introduces the notion of controlling forcing sets, and gives tight bounds on the size of controlling sets for the Miura-ori crease pattern.
For a given crease pattern $C$, 
a subset $F \subseteq C$ is \emph{controlling forcing}, or simply \emph{controlling}, if it is forcing for every mountain-valley assignment $\mu$ on $C$. 
\begin{figure}[t]
\centering
\includegraphics[width=\linewidth]{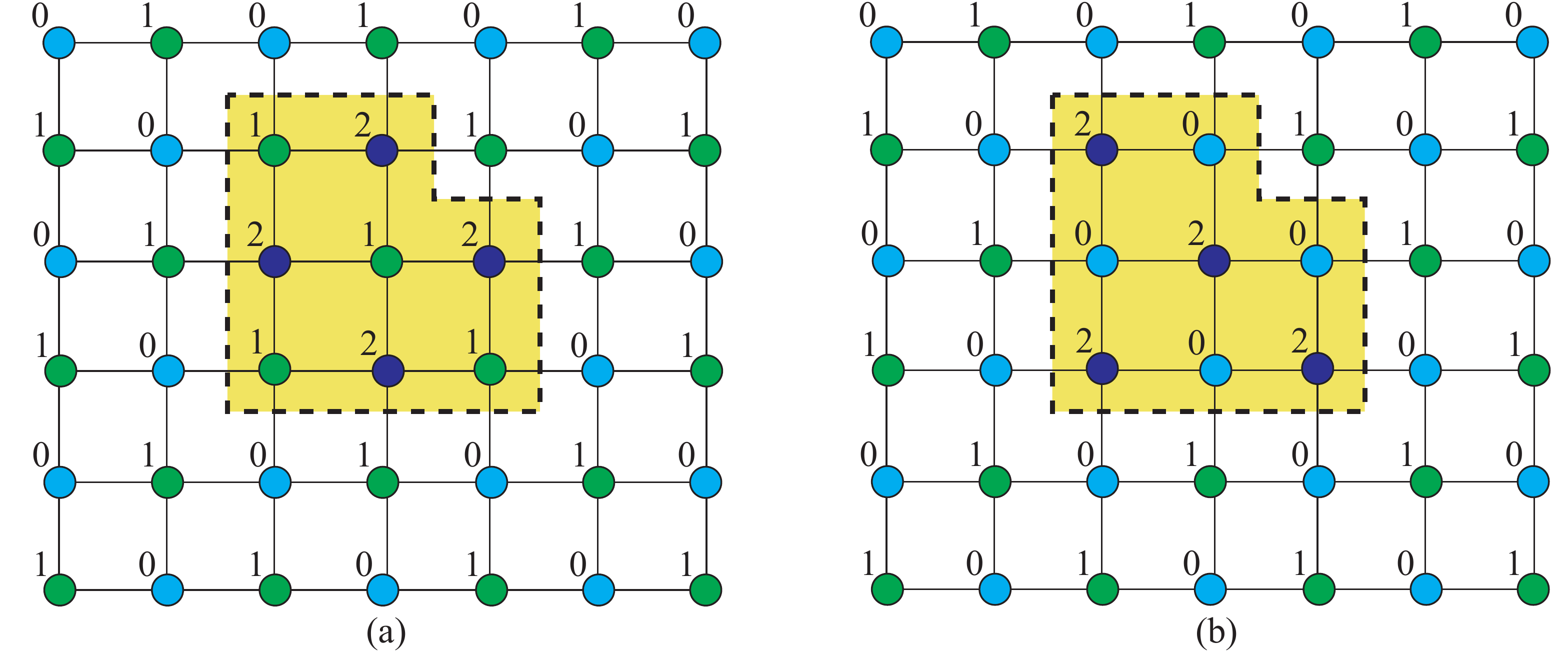}
\caption{The shaded region marks a single connected component induced by the input set $S$. (a) A checkerboard coloring with every other vertex colored 1. (b) An alternate checkerboard coloring, with every other vertex colored 0. 
A controlling set must include an edge connecting the shaded component to the rest of the graph, otherwise it cannot distinguish between colorings (a) and (b). 
}
\label{fig:checkerboard}
\end{figure}
%
As we show in this section, controlling sets for Miura-ori must be significantly larger than forcing sets for particular MV-assignments.

\begin{theorem}
\label{thm:spanningtree}
A subset $S$ of the Miura-ori creases is controlling if and only if it contains a spanning tree of the dual grid graph. 
\end{theorem}
\begin{proof}
In one direction, if $S$ contains the creases corresponding to the edges of a spanning tree in the dual grid graph, it is easy to show that $S$ is forcing by propagating colors along the tree from the fixed color in the top left corner.

In the other direction, if the grid graph edges corresponding to $S$ do not span the entire grid graph, then we show that there exist two valid checkerboard colorings that cannot be distinguished by $S$ and therefore $S$ is not a forcing set.
\figref{checkerboard} shows an example where $S$ induces two or more connected components in the grid graph, one of which consists of the vertices in the shaded L-shaped region.     
For the first checkerboard pattern, we assign every other vertex of the grid graph the color 1. 
Then we pick an arbitrary component and assign its remaining vertices 
the color 2, while everywhere else we assign the remaining vertices the color 0, as illustrated in \figref{checkerboard}a.
For the second pattern, we change the 1-2 checkerboard pattern of the selected component to 2-0 (see \figref{checkerboard}b). This preserves the validity of the coloring and cannot be distinguished by a supposed forcing set that includes no edges between the component's vertices
and the rest of the graph.    
\end{proof}

\section{Conclusions}
Inspired by the need to reduce the cost of pre-programmed self-folding mechanisms, we study forcing sets for the classical Miura-ori folding pattern. We present efficient algorithms for determining a minimum forcing set of a Miura-ori map (Sections~\ref{sec:standardmiura} and~\ref{sec:nonstandard}), and for extending a given MV assignment to a forcing set, if possible~(\autoref{sec:forcingsetcharacterization}). We also  introduce the notion of controlling forcing sets, and give tight bounds on the size of controlling sets for the Miura-ori crease pattern (\autoref{sec:tightcontrol}).

Our work opens several new questions, chief among which stands the question of global flat-foldability.  It would be interesting to determine if the MV-assignments induced by some of the algorithms in this paper are globally flat-foldable or not. Extending the results beyond the confines of the Miura-ori crease pattern remains open.

\bibliographystyle{abuser}
\bibliography{forcebib}

\end{document}